\documentclass[journal, twocolumn]{IEEEtran}

\IEEEoverridecommandlockouts

\usepackage[utf8]{inputenc}
\usepackage[T1]{fontenc}

\usepackage{amsmath}        
\usepackage{amssymb}        
\usepackage{amsfonts}       
\usepackage{amsthm}         

\usepackage[scr=rsfs]{mathalpha}

\usepackage{graphicx}
\usepackage{xcolor}
\usepackage{textcomp}
\usepackage{subfigure}
\usepackage{caption}
\usepackage{threeparttable}
\usepackage{multirow}
\usepackage{array}

\usepackage{diagbox} 
\usepackage{footnote}

\usepackage[ruled,linesnumbered]{algorithm2e}

\usepackage{cite}
\usepackage{url}

\usepackage{setspace}
\usepackage{footnote}
\usepackage{ifthen}

\usepackage{mathtools}

\newtheorem{theorem}{Theorem}  
\newtheorem{lemma}{Lemma}      

\newtheorem{proposition}{Proposition}

\newtheorem{remark}{Remark}

\newtheorem{assumption}{Assumption}

\begin{document}

\title{
Enabling Safety-Critical Wireless Communications via Safe Reinforcement Learning
}
\IEEEaftertitletext{\vspace{-2.2\baselineskip}}
\author{ \IEEEauthorblockN{Haoran Peng,~\IEEEmembership{Member,~IEEE},~Tong Wu,~\IEEEmembership{Member,~IEEE},~Hang Liu,~\IEEEmembership{Member,~IEEE},~Weijia Zheng,~\IEEEmembership{Graduate Student Member,~IEEE},~Ying-Jun Angela Zhang,~\IEEEmembership{Fellow,~IEEE} and~Anna Scaglione,~\IEEEmembership{Fellow,~IEEE}}
\thanks{
Haoran Peng, Weijia Zheng, and Ying-Jun Angela Zhang are with the Department of Information Engineering, The Chinese University of Hong Kong, Shatin, N.T., Hong Kong S.A.R. (email:~\{hrpeng, zw023, yjzhang\}@ie.cuhk.edu.hk). 
Tong Wu is with the Department of Electrical and Computer Engineering, University of Central Florida, Orlando, FL USA (email:~tong.wu@ucf.edu). Hang Liu is with the State Key Laboratory of Internet of Things for Smart City and the Department of Electrical and Computer Engineering, University of Macau, Macao S.A.R. (email: hangliu@um.edu.mo). Anna Scaglione is with the Department of Electrical and Computer Engineering, Cornell Tech, Cornell University, NY 10044 USA (email: as337@cornell.edu).
}
}

\markboth{}%
{Shell \MakeLowercase{\textit{et al.}}: Bare Demo of IEEEtran.cls for IEEE Communications Society Journals}
\maketitle

\begin{abstract}
Ensuring strict safety guarantees is the paramount challenge for emerging 5G/6G wireless systems, particularly as they increasingly govern mission-critical applications ranging from autonomous UAV swarms to industrial automation.
While deep reinforcement learning (DRL) offers a promising solution for complex resource allocation, standard algorithms frequently violate essential  constraints, such as QoS mandates and power limits, posing unacceptable risks of system failure and regulatory non-compliance.
We propose Safe-Deep Q-Learning, a novel algorithm that simultaneously addresses all three challenges: it handles mixed-integer nonconvex problems by approximating the Q-function, adapts to stochastic dynamics, and enforces dual-timescale constraints using integrated Lagrangian methods. 
Our framework features adaptive penalty scaling and constraint violation tracking, specifically tailored for wireless environments, and is designed to operate in both distributed and centralized architectural modes. 
We prove convergence to optimal constraint-satisfying policies under mild conditions and demonstrate robustness through dual variable stabilization.
Validation on unmanned aerial vehicle (UAV) swarm control network and post-disaster emergency communications applications shows that Safe-Deep Q-Learning achieves stringent adherence to safety bounds with near-zero violation rates, significantly outperforming existing constrained RL baselines, establishing its effectiveness for safety-critical wireless deployments.
\end{abstract}


\section{Introduction}
\label{Sec: Intro}
The rapid evolution towards Beyond 5G (B5G) and 6G ecosystems heralds a paradigm shift from simple connectivity to the ``Internet of Everything”, where ubiquitous intelligence permeates a vast, heterogeneous network of devices~\cite{11017513}. 
This transition catalyzes the emergence of unprecedented, hyper-dynamic scenarios, such as high-mobility UAV swarms, autonomous vehicular platoons, and space-air-ground integrated networks, that operate far beyond the design parameters of legacy systems, as illustrated in Fig.~\ref{Fig:bg}~\cite{11031135}. 
Consequently, traditional communication architectures, originally engineered for static topologies and predictable traffic patterns, are proving fundamentally inadequate to meet the stringent safety and reliability constraints imposed by these highly volatile environments~\cite{10974617,10443575}.
These emerging applications operate under multiple interdependent constraints, including limited energy budgets, spectrum masks, interference thresholds, and latency requirements that must be simultaneously satisfied across dynamic and uncertain wireless environments~\cite{10970466, 10700695}.


The inherent complexity of these applications stems from three fundamental challenges that render classical optimization inadequate.
First, the resource allocation problem corresponds to a \textit{high-dimensional mixed-integer nonconvex optimization} problem, where discrete decisions (e.g., user association, spectrum access) and continuous variables (e.g., beamforming, power control, bandwidth allocation, route planning, steering control) are intimately coupled across spatially distributed entities, spanning user terminals, edge servers, and space-air-ground integrated nodes \cite{8854339,10051712}.
In such distributed heterogeneous environments, traditional iterative solvers like Alternating Optimization (AO) often fail to converge to feasible solutions within the ultra-low latency coherence times of wireless channels.
Second, these systems exhibit \textit{stochastic dynamics} where channel conditions, traffic patterns, and energy availability evolve unpredictably, invalidating the assumptions of static optimization \cite{10746553}. 
This volatility is further exacerbated by high-mobility devices, such as UAVs, aircraft, and autonomous vehicles, which not only increase  environmental fluctuations but also impose significantly stricter latency requirements to maintain reliable connectivity.
Third, they demand  constraint satisfaction in two different time scales, requiring the simultaneous enforcement of  \textit{instantaneous hard} bounds (e.g., collision obstacles, velocity, path and mission planning, and emergency QoS,  sub-millisecond URLLC latency) for immediate safety, alongside \textit{cumulative soft} constraints (e.g., average energy budgets, QoS fairness) to sustain the networking capabilities.
The diverse dual-time scale safety-critical constraints for some typical scenarios are summarized in Table~\ref{tab:constraints}.


We exemplify these challenges through a representative mission-critical multi-UAV network.
In this environment, high-dimensional decision variables related to 3D trajectories and spectrum allocation are distributed across distinct aerial agents, yet remain tightly coupled through mutual interference and physical collision risks. 
The system demands the simultaneous satisfaction of instantaneous hard constraints, such as collision avoidance and ultra-reliable signaling, alongside cumulative energy budgets.
This complex coupling among spatially dispersed entities makes rigorous safety assurance intractable for conventional static optimization.

Classical approaches, which typically formulate these as mixed integer nonlinear programming (MINLP) problems, represent static optimization frameworks that fail to address this confluence of challenges. These methods struggle particularly in hostile or rapidly changing propagation conditions, frequently operating in infeasible regions or violating critical safety bounds during transient periods~\cite{9683088}. Moreover, the temporal-correlation nature of wireless radio environment, where current decisions affect future states through battery depletion, queue evolution, and thermal dynamics, necessitates multi-stage optimization under uncertainty that static methods cannot handle~\cite{11031135}.
This fundamental mismatch between the dynamic, uncertain nature of modern wireless systems and the limitations of traditional optimization has motivated the need for online control frameworks, particularly learning-based approaches that can adapt to stochastic dynamics while providing mathematical guarantees for constraint satisfaction.
Such a paradigm is for deploying wireless systems in safety-critical real-world applications~\cite{10877725}.

\begin{figure*}[!tbp]
\centerline{\includegraphics[width=.975\textwidth]{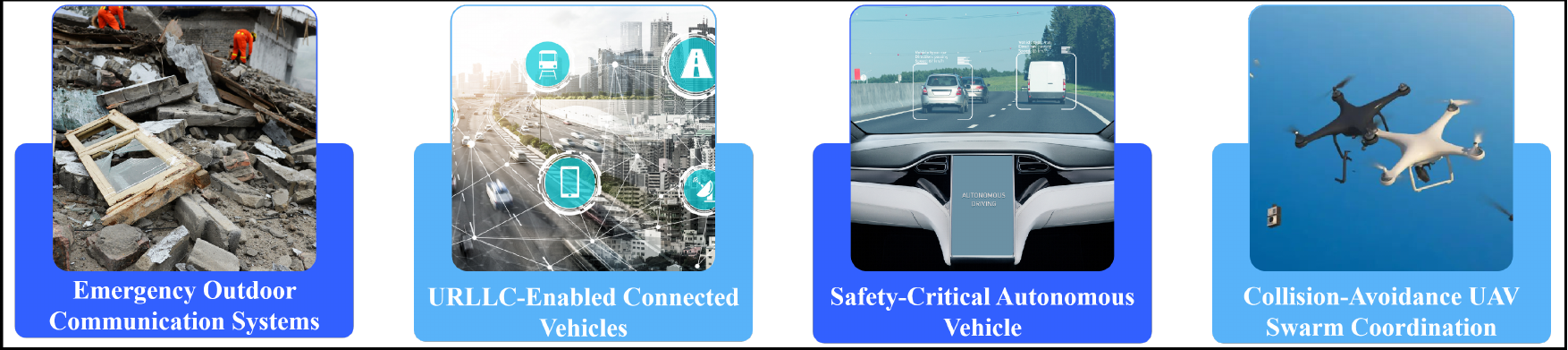}}
\caption{
The background of the safety-critical communication applications.
}
\label{Fig:bg}
\end{figure*}

\begin{table*}[!tbp]
\caption{Safety-critical constraints in real-world wireless communication systems}
\centering
\begin{tabular}{l|c|c|c|c}
\hline\hline
\textbf{Category} & \textbf{Emergency Commun.} & \textbf{Connected Vehicles} & \textbf{Autonomous Driving} & \textbf{UAV Swarm} \\
\hline\hline
\textbf{Hard Constraints}
& Transmit power limits$^{\ddagger}$; & End-to-end latency$^{\dagger}$;& Collision avoidance$^{\ddagger}$;& DAA\textsuperscript{*} delivery$^{\dagger}$; \\
& Commun. reliability$^{\dagger}$;& Commun. integrity$^{\dagger}$;& Real-time perception$^{\dagger}$; & Energy budget$^{\dagger}$;\\
\hline
\textbf{Soft Constraints}
& Spectrum usage efficiency$^{\dagger}$& Broadband throughput$^{\dagger}$& Positioning accuracy$^{\dagger}$ &Formation stability$^{\ddagger}$ \\
\hline\hline
\end{tabular}
\vspace{4pt}
\raggedright
\footnotesize{
\textsuperscript{*}DAA: \textbf{D}etect \textbf{a}nd \textbf{A}void message. \\
$^{\dagger}$ Constraints optimized and managed by the proposed framework (e.g., via Safe-Deep Q-Learning). \\
$^{\ddagger}$ Constraints enforced by deterministic base-layer mechanisms (e.g., safety shields or hardware limits).
}
\label{tab:constraints}
\end{table*}

\subsection{Related Work}
\label{relatedwork}
To address these fundamental challenges, researchers have developed various approaches that partially tackle the dynamic and constraint-aware aspects of wireless resource allocation. While no existing method comprehensively solves all critical issues simultaneously, we categorize the current literature into three main directions based on their primary focus:
\subsubsection{Dynamic Programming}
Existing solutions model dynamic resource allocation as Markov Decision Processes (MDPs) to capture the sequential decision-making nature. 
However, standard MDPs struggle in wireless environments due to their inability to enforce hard physical constraints (e.g., power budgets, fronthaul capacity) and critical service requirements like URLLC~\cite{10529729}. 
These challenges are particularly critical in scenarios requiring strict multi-constraint satisfaction, as violations can lead to catastrophic system failures, such as physical collisions in UAV swarms due to latency breaches or loss of control in autonomous vehicles.
Constrained MDPs (CMDPs) overcome these limitations by embedding domain-specific engineering constraints within state-evolving feasible actions and cost thresholds~\cite{11074719}. Dynamic programming (DP) has proven effective in various wireless applications within CMDP frameworks, spanning resource management and semantic communication. For instance, \cite{9968223} applied DP to optimize cache placement in cloud radio access networks (C-RANs) featuring hybrid millimeter-wave/microwave fronthaul links, directly addressing communication constraints for efficient data delivery. Similarly, \cite{10269137} developed an adaptive dynamic programming (ADP) approach to optimize source selection and resource allocation in wireless-powered relay networks while adhering to stringent energy constraints. Furthermore, \cite{9770190} applied DP to context-based semantic communication, proposing a framework that optimizes message transmission by considering both semantic relevance and communication constraints.

Despite these advances, the curse of dimensionality makes exact DP solutions computationally prohibitive for large-scale systems with high-dimensional state spaces. Moreover, CMDPs typically require known system models and struggle with the nonconvex, mixed-integer nature of many wireless optimization problems, often necessitating approximations that may compromise constraint satisfaction guarantees.

\subsubsection{Deep Reinforcement Learning (DRL)}
DRL has emerged as the dominant paradigm for addressing high-dimensional and mixed-integer wireless control tasks, including dynamic resource allocation, beamforming, and distributed network management. Recent studies apply off-the-shelf DRL algorithms with domain-specific network architectures to large-scale resource allocation problems, demonstrating strong empirical gains in throughput, latency, and energy metrics compared to classical heuristics~\cite{10605908}.
Several notable applications showcase DRL's effectiveness in wireless systems. \cite{10666879} proposed a multi-agent DRL approach for two-timescale resource allocation in network slicing for Vehicle-to-Everything (V2X) communications, optimizing resources across fast and slow timescales while adapting to dynamic network loads in large-scale vehicular networks. \cite{8771176} applied DRL for dynamic computation offloading in wireless-powered mobile edge computing networks, optimizing computation efficiency and energy consumption in resource-constrained environments. Similarly, \cite{10636057} developed a DRL-based joint beamforming design for active reconfigurable intelligent surfaces (RIS)-enabled cognitive multicast systems, focusing on optimizing beamforming while accounting for wireless channel dynamics.

Nevertheless, conventional DRL solutions face critical limitations in safety-critical wireless applications.
Importantly, they typically lack formal safety guarantees and can violate operational constraints during exploration phases in real-time networks. The exploration-exploitation trade-off, inherent in DRL, can lead to constraint violations that are unacceptable in mission-critical scenarios. Furthermore, standard DRL methods provide no theoretical convergence guarantees for constrained problems, making them unsuitable for applications requiring strict constraint adherence throughout the learning process.

\subsubsection{Safe Learning}
In recent years, various safe learning approaches have been proposed that incorporate safety constraints directly into the learning process, enabling deployment of learning-based control systems in safety-critical applications~\cite{10675394}.
Methods such as Lyapunov-based actor-critic and Lagrangian relaxation-based RL have been adapted for radio resource allocation under safety requirements~\cite{10675394,10943268}.
Learning process modification techniques embed safety constraints directly into policy updates through various mechanisms. 
Lyapunov-based approaches optimize offloading decisions and energy consumption in mobile edge computing systems while ensuring queue stability~\cite{10974474}. \cite{11087615} proposed a two-stage framework for safe RL that minimizes power consumption in wireless networked control systems while satisfying peak age of information (AoI) violation probability constraints, using a ``teacher” network to guide the RL training process under safety constraints.
A key advancement in safe RL involves applying stochastic dynamic models to real-world safety-critical decision-making scenarios. Lagrangian-based RL effectively addresses CMDP challenges by combining dual-ascent adaptability for constraint thresholds with quadratic penalty enforcement for adherence to feasible regions. 
This approach provides mathematical rigor with asymptotic feasibility guarantees while ensuring sample efficiency.

However, despite these advances, generic safe learning frameworks often fail to provide the strict, instantaneous safety guarantees required for mission-critical wireless deployments.
The above approach typically incorporates constraints as regularization (penalty) terms in the objective to promote smoother optimization and convergence.
However, this formulation yields a multi-objective trade-off and does not guarantee satisfaction of \emph{instantaneous hard} constraints.
As a result, constraint violations may still occur during learning or deployment, which can be unacceptable in zero-tolerance safety settings, such as collision avoidance or stringent outage constraints.

\subsection{Contributions}
This paper advances safe learning for wireless networks by introducing a novel primal-dual RL framework that provides mathematical guarantees for hard constraint satisfaction during operation.
Crucially, rather than redundantly learning fundamental physical rules from scratch, this study introduces a hierarchical shielded architecture. 
Operating atop mature, deterministic control baselines (acting as a physical safety shield), the proposed approach serves as a higher-order optimization layer.
The core contribution lies in synergizing these strict physical boundaries with high-dimensional networking objectives, such as joint spectrum and power allocation, thereby guaranteeing absolute operational safety while maximizing communication performance.
While conventional control paradigms generally isolate physical safety (e.g., collision avoidance or hardware limits) from data transmission, our methodology explicitly couples these domains. Physical states are dynamically optimized to serve underlying networking objectives. This integration creates a highly constrained problem where communication demands (e.g., spectrum allocation and interference mitigation) must be jointly satisfied within strict physical operational boundaries.
While the proposed framework accommodates various learning architectures, we specifically instantiate it as Safe-Deep Q-Learning to efficiently resolve the combinatorial integer variables inherent to wireless systems, such as spectrum allocation. Crucially, unlike on-policy methods that suffer from severe gradient variance and instability when exposed to fluctuating Lagrangian penalties, our off-policy value-based approach leverages experience replay to smoothly absorb these penalty spikes.
This algorithmic design guarantees robust convergence in strictly constrained discrete action spaces.
The key contributions are organized into methodology and application aspects:
\begin{itemize}

\item \textbf{Constraint-Aware Problem Modeling:} We model wireless resource allocation as a constrained MDP that enables simultaneous handling of discrete decisions, uncertain dynamics, and multi-timescale safety requirements.

    \item \textbf{Safe-Deep Q-Learning Algorithm:} 
    We propose a primal-dual deep Q-learning, featuring constraint violation tracking and adaptive penalty scaling tailored for wireless environments.

    \item \textbf{Theoretical Guarantees:} We prove both convergence to optimal policies and constraint satisfaction guarantees, ensuring feasibility throughout learning while achieving asymptotic optimality.

    \item \textbf{Practical Validation:} We apply Safe-Deep Q-Learning to two safety-critical applications, multi-UAV wireless communications and post-disaster emergency communication networks, to validate its effectiveness in solving mixed-integer path planning under stochastic channels with instantaneous safety constraints and dynamic resource allocation under uncertain arrivals with dual-timescale requirements, respectively.
\end{itemize}

\subsection{Paper Organization and Notations}
\textbf{Key Notations:} Throughout this paper, we use $\mathcal{S}$ and $\mathcal{A}$ to denote the state and action spaces, respectively. The action space comprises discrete variables $\mathbf{x}_t \in \{0,1\}^{N_d}$ and continuous variables $\mathbf{y}_t \in \mathbb{R}^{N_c}$. The policy $\pi: \mathcal{S} \to \mathcal{P}(\mathcal{A})$ maps states to action distributions. Constraint sets are indexed by $\mathcal{I}$ (cumulative), $\mathcal{J}$ (equality), and $\mathcal{K}$ (inequality) with corresponding cost functions $c_i(\cdot)$, $e_j(\cdot)$, and $g_k(\cdot)$. Lagrange multipliers are denoted by $\boldsymbol{\lambda} = \{\lambda_i\}_{i \in \mathcal{I}}$, $\boldsymbol{\mu} = \{\mu_j\}_{j \in \mathcal{J}}$, and $\boldsymbol{\nu} = \{\nu_k\}_{k \in \mathcal{K}}$ with penalty parameters $\rho_i, \rho_j, \rho_k \geq 0$. Bold symbols represent vectors (e.g., $\mathbf{H}_t$ for channel gains, $\mathbf{Q}_t$ for queue states), while calligraphic fonts denote sets and spaces.

\textbf{Paper Organization:}
The remainder of this paper is structured as follows. Section~\ref{sec:gen_framework} establishes the mathematical foundation by formulating the safety-critical wireless resource allocation problem as a CMDP. Section~\ref{Sec:SafeRL} develops the proposed Safe-Deep Q-Learning framework, detailing the augmented Lagrangian mechanism and the dual-timescale algorithmic design. The theoretical analysis regarding convergence, optimality, and constraint satisfaction is rigorously provided in Section~\ref{sec:theoretical_analysis}. The proposed method is validated through two distinct case studies: mission-critical multi-UAV networks in Section~\ref{subsec:uav_system} and post-disaster emergency communication networks in Section~\ref{sec:studycase2}. Finally, Section~\ref{conclusions} concludes the paper with key findings.

\section{Problem Formulation: Safe Control in Wireless Systems}
\label{sec:gen_framework}

Building upon the challenges identified in Section~\ref{Sec: Intro}, we formulate the safe wireless resource allocation problem as a CMDP that addresses all three fundamental issues: mixed-integer nonconvex optimization, stochastic dynamics, and dual-timescale constraint satisfaction. 
This formulation provides the mathematical foundation for developing Safe-Deep Q-Learning algorithms that guarantee constraint satisfaction while optimizing system performance in dynamic wireless environments.

\subsection{Safety-Constrained Wireless Resource Optimization}
\label{subsec:constrained_wireless}
Resource allocation in emerging safety-critical wireless systems transcends traditional throughput maximization, serving as the fundamental control layer that governs physical system integrity.
This process involves high-dimensional discrete decisions (e.g., user selection, channel assignment) and continuous variables (e.g., power levels, beamforming, mobility control) that must be optimized within strict physical boundaries.
Unlike legacy best-effort networks, these decision variables are inextricably coupled with hard safety limits, such as hardware thermal thresholds, kinetic collision avoidance distances, and interference tolerance levels, where violations risk catastrophic operational failures.
We formulate this as a general dynamic optimization problem that captures the temporal-correlation nature of wireless resource allocation under strictly enforced physical uncertainties.

At each time step $t$, the system state $s_t \in \mathcal{S}$ encapsulates the current wireless environment including channel gains $\mathbf{H}_t$, queue states $\mathbf{Q}_t$, energy levels $\mathbf{E}_t$, and user demands $\mathbf{D}_t$. 
The action space $\mathcal{A}$ comprises both discrete variables $\mathbf{z} \in \mathcal{Z}$, where $\mathcal{Z}=\{\zeta_1, \dots, \zeta_L\}^{N_d}$ is a finite set of discrete feasible values (e.g., MCS levels, subchannel indices, or codebook entries), and continuous variables $\mathbf{y} \in \mathbb{R}^{N_c}$ (e.g., power allocation $p_u \in [0, P_{\max}]$, bandwidth $b_{u,c} \in [0, B_{\max}]$).
Thus, each action $a_t = (\mathbf{x}_t, \mathbf{y}_t)$ represents a joint discrete-continuous decision.

The wireless resource allocation problem can be cast in the following constrained optimization formulation:
\begin{align}
&\text{(P1):}~\max_{\pi} \mathbb{E}_{\pi}\left[ \sum_{t=0}^{\infty} \gamma^t r(s_t, a_t) \right] \label{eq:wireless_opt} \\
&\text{s.t.}~\mathbb{E}_{\pi}\left[ \sum_{t=0}^{\infty} \gamma^t c_i(s_t, a_t) \right] \leq d_i, \quad \forall i \in \mathcal{I}, \label{eq:cumulative_constr}  \\
&\quad \quad g_k(s_t, a_t) \leq 0, \quad \forall k \in \mathcal{K}, \quad \forall t \geq 0, \label{eq:instant_constr}  \\
&\quad \quad e_j(s_t, a_t) = 0, \quad \forall j \in \mathcal{J}, \quad \forall t \geq 0, \label{eq:equality_constr} \\
&\quad \quad \mathbf{z}_t \in \mathcal{Z}, \quad \mathbf{y}_t \in \mathcal{Y}(s_t), \quad \forall t \geq 0, \label{eq:domain_constr}
\end{align}
where $\pi: \mathcal{S} \to \mathcal{P}(\mathcal{A})$ denotes the policy mapping states to action distributions, $r(s_t, a_t)$ represents the system reward (e.g., sum-rate, energy efficiency), and $\gamma \in [0,1)$ is the discount factor. Constraint~\eqref{eq:cumulative_constr} enforces cumulative resource limits with cost functions $c_i(\cdot)$ and budgets $d_i$, while~\eqref{eq:instant_constr} ensures instantaneous safety requirements through inequality functions $g_k(\cdot)$. 
Additionally, \eqref{eq:equality_constr} imposes strict operational equality conditions via $e_j(\cdot)$, such as orthogonal spectrum assignment or exact resource matching. 
The feasible continuous action space $\mathcal{Y}(s_t)$ depends on the current state, capturing state-dependent constraints such as available power and bandwidth.

The mathematical structure of Problem (P1) embodies the fundamental challenges identified in Section~\ref{Sec: Intro}. 
First, spatially distributed users are inextricably coupled through mutual interference and physical safety requirements embedded in $g_k(\cdot)$, transforming local mixed-integer decisions into a high-dimensional combinatorial problem.
Second, the absence of real-time global information at distributed entities exacerbates the stochasticity of state transitions $P(s_{t+1}|s_t, a_t)$, rendering accurate prediction of environmental dynamics infeasible.
These complexities necessitate the rigorous formulation of dual-timescale safety enforcement, the third fundamental challenge, which is detailed in Section~\ref{subsec:wireless_constraints}.

\subsection{Safety Constraints}
\label{subsec:wireless_constraints}

The constraints in Problem (P1) directly correspond to the operational requirements and safety mandates in wireless communication systems. We categorize these constraints based on their temporal characteristics and criticality levels, as introduced in the related work discussion.

\textbf{Instantaneous Hard Constraints} (e.g., Constraints~\eqref{eq:instant_constr} and \eqref{eq:equality_constr}) enforce strict per-timestep compliance for safety-critical parameters:
\begin{itemize}
\item \textit{Power Constraints (Inequality):} Individual transmit power limits $p_u(s_t, a_t) \leq P_{\max,u}$ and total power budgets $\sum_{u} p_u(s_t, a_t) \leq P_{\text{total}}$ prevent hardware damage and ensure regulatory compliance.
\item \textit{Velocity/Path Control (Inequality):} 
Velocity and path planning constraints $\|\mathbf{v}_u(s_t, a_t)\|_2 \leq V_{max}$ ensure flight stability and prevent agents from exiting designated operational zones (geofencing), where $\mathbf{v}_u$ denotes the instantaneous velocity vector.
\item \textit{Physical Safety (Inequality):} 
Collision avoidance is paramount for mobile agents. For UAV swarms, inter-agent separation $\|\mathbf{p}_i(s_t, a_t) - \mathbf{p}_j(s_t, a_t)\|_2 \geq d_{\min}$ ensures safety in 3D space; for connected vehicles, obstacle avoidance requires $\|\mathbf{p}_i(s_t, a_t) - \mathbf{p}_\mathcal{O}\|_2 \geq d_{obs}$ to maintain safe headway from static or dynamic obstacles $\mathcal{O}$.
\item \textit{Spectrum Constraints (Inequality):} Channel assignment constraints $\sum_{u} z_{u,c} \leq 1$ prevent multiple users from using the same frequency band simultaneously.
\item \textit{Intrinsic Hardware Limitations (Equality):} 
Emerging B5G/6G heterogeneous hardware substrates are subject to strict physical operational boundaries defined by their underlying circuitry. 
For instance, in RIS-aided architectures, the absence of active RF chains and power amplifiers compels the passive reflecting elements to strictly adhere to power conservation principles, necessitating the unit-modulus equality constraint $|e^{j\theta_m}| - 1 = 0$.
\item \textit{Resource Exclusion (Equality):} Binary association constraints, denoted as $e_j(s_t,a_t) = \sum_{u}z_{u,c} - 1 = 0$, ensure that each orthogonal subchannel is assigned to exactly one user (or idle), strictly enforcing spectrum exclusivity.
\end{itemize}

\textbf{Cumulative Resource Constraints} (e.g., Constraint~\eqref{eq:cumulative_constr}) govern finite-horizon resource budgeting:
\begin{itemize}
\item \textit{Energy Budget:} Total energy consumption limits $\sum_{t=0}^{T} \gamma^t \sum_u p_u(s_t, a_t) \leq E_{\text{budget}}$ for battery-powered devices.
\item \textit{Spectrum Licensing:} Licensed spectrum usage quotas $\sum_{t=0}^{T} \gamma^t \sum_c B_{c}(s_t, a_t) \leq B_{\text{quota}}$ ensure regulatory compliance.
\item \textit{QoS Guarantees:} Long-term throughput requirements $\sum_{t=0}^{T} \gamma^t R_u(s_t, a_t) \geq R_{\min,u}$ for service level agreements.
\item \textit{Service Reliability:} 
Cumulative outage constraints ${\mathbb{E}\left[\sum_{t=0}^{T}\gamma^{t}\mathbb{I}\text{SINR}_{u}(t) \leq \Gamma_{th}\right]}\leq \epsilon_{out}$
limits the frequency of connection failures, ensuring long-term service continuity for mission-critical data flows.
\end{itemize}

\section{Safe Reinforcement Learning for Wireless Systems}
\label{Sec:SafeRL}
To address these fundamental challenges, we develop Safe-Deep Q-Learning that integrates constraint-aware action selection with adaptive penalty mechanisms. 
In this framework, the augmented Lagrangian formulation is directly embedded into the temporal difference learning process, where linear dual variables and quadratic penalty terms associated with constraints are incorporated into the Q-function update target. 
By reshaping the Bellman operator with these penalty signals, the algorithm learns safety-aware value estimates that inherently penalize infeasible transitions, guiding the policy to maximize rewards while strictly adhering to the feasible region.
\subsection{Safe RL Formulation and Algorithmic Challenges}
\label{subsec:safe_rl_formulation}
To address the complex dynamics inherent in Problem (P1), we abstract the specific physical wireless resource allocation problem into a generalized reinforcement learning framework.
Unlike standard RL that focuses solely on reward maximization, safe RL extends the framework to handle CMDPs where safety requirements must be maintained during both learning and deployment phases.

Building upon Problem (P1), we formalize the safe RL framework as a CMDP $\langle \mathcal{S}, \mathcal{A}, P, r, \mathbf{c}, \mathbf{e},\mathbf{g}, \gamma \rangle$, where the safe RL problem seeks a policy $\pi: \mathcal{S} \to \mathcal{P}(\mathcal{A})$ that maximizes expected discounted rewards while satisfying safety constraints:
\begin{align}
&\text{(P2):}~\max_{\pi} ~ V_r(\pi) = \mathbb{E}_{\pi}\left[ \sum_{t=0}^{\infty} \gamma^t r(s_t, a_t) \right] \label{eq:safe_rl_objective} \\
&\text{s. t.} \quad  V_{c_i}(\pi) = \mathbb{E}_{\pi}\left[ \sum_{t=0}^{\infty} \gamma^t c_i(s_t, a_t) \right] \leq d_i, ~\forall i \in \mathcal{I}, \label{eq:safe_rl_cumulative} \\
&\quad\quad g_k(s_t, a_t) \leq 0, \quad \forall k \in \mathcal{K},~\forall t \geq 0, \label{eq:safe_rl_instantaneous}\\
&\quad\quad e_j(s_t, a_t) = 0, \quad \forall j \in \mathcal{J}, ~\forall t \geq 0,\label{eq:safe_rl_equality} 
\end{align}
where $\gamma \in [0,1)$ is the discount factor, $d_i \geq 0$ is the cost threshold.
The functions 
$e_j:\mathcal{S}\times\mathcal{A}\rightarrow\mathbb{R}$
 and $g_k: \mathcal{S} \times \mathcal{A} \to \mathbb{R}$
 define instantaneous equality (e.g., spectrum logical exclusion, unit-modulus phase shifts) and inequality safety boundaries, respectively.
This formulation directly corresponds to the constraint structure in Problem (P1), where constraints~\eqref{eq:safe_rl_cumulative}, ~\eqref{eq:safe_rl_instantaneous} and \eqref{eq:safe_rl_equality} map to cumulative constraints~\eqref{eq:cumulative_constr}, instantaneous inequality constraints~\eqref{eq:instant_constr} and instantaneous equality constraints \eqref{eq:equality_constr}, respectively.

The application of safe RL to wireless systems presents fundamental algorithmic challenges that extend beyond standard constrained RL methods. First, \textit{dual-timescale constraint handling} requires distinct treatment: cumulative constraints~\eqref{eq:safe_rl_cumulative} can be addressed through Lagrangian relaxation and dual variable updates, while instantaneous constraints~\eqref{eq:safe_rl_instantaneous} demand projection methods or barrier functions to pursue strict feasibility at every timestep.
Second, the \textit{mixed-integer nature} of wireless decisions, where discrete variables $\mathbf{z}_t \in \{0,1\}^{N_d}$ couple with continuous variables $\mathbf{y}_t \in \mathcal{Y}(s_t)$, creates combinatorial complexity that off-the-shelf DRL architectures cannot effectively manage. These methods generally lack the structural mechanisms to perform simultaneous joint optimization over hybrid action spaces, often resorting to sub-optimal relaxation or discretization techniques.

Third, \textit{nonconvex feasible regions} arising from interference-coupled constraints create highly irregular action spaces $\mathcal{A}_{\text{safe}}(s_t) = \{a \in \mathcal{A} : g_k(s_t, a) \leq 0, \forall k \in \mathcal{K}, \text{and}~e_j(s_t,a)=0, \forall j\in\mathcal{J}\}$ that fundamentally challenge traditional convex optimization and unconstrained RL approaches.

To address these challenges, we propose a novel primal-dual RL framework that provides mathematical guarantees for hard constraint satisfaction. 
While this framework is theoretically extensible to arbitrary RL algorithms, we specifically instantiate it as Safe-Deep Q-Learning in this work. 
This design choice explicitly addresses the inherent combinatorial and mixed-integer nature of wireless resource allocation (e.g., discrete spectrum access), exploiting the superior stability and convergence speed of value-based methods in discrete action spaces~\cite{mnih2015human}.

\subsection{Preliminaries on Deep Q-Learning}
The Deep Q-Network (DQN) approximates the optimal action-value function $Q^*(s, a)$ using a non-linear neural network parameterized by $\theta$. The standard Q-learning objective seeks to minimize the mean squared temporal difference (TD) error between the current value estimate and the target value derived from the Bellman optimality equation. The loss function is defined as:
\begin{align}
\mathcal{L}(\theta) = \mathbb{E}_{(s_t, a_t, r_t, s_{t+1}) \sim \mathcal{D}} \left[ \left( y_t - Q_\theta(s_t, a_t) \right)^2 \right],
\end{align}
where $\mathcal{D}$ denotes the experience replay buffer employed to decorrelate training samples. The target value $y_t$ is computed using a separate target network with parameters $\hat{\theta}$, which are periodically synchronized with $\theta$:
\begin{align}
y_t = r(s_t, a_t) + \gamma \max_{a' \in \mathcal{A}} Q_{\hat{\theta}}(s_{t+1}, a').
\end{align}
While effective for unconstrained reward maximization, standard DQN lacks the intrinsic mechanism to enforce the strict safety constraints defined in Problem (P1), necessitating the integration of the augmented Lagrangian framework detailed in the subsequent subsections.

\subsection{Augmented Lagrangian Safe RL Framework}
\label{subsec:alm_framework}

Building upon the safe RL formulation from Section~\ref{subsec:safe_rl_formulation}, we develop an augmented Lagrangian framework that transforms the constrained problem into an unconstrained framework while maintaining strict constraint satisfaction guarantees. 
It is worth noting that in many wireless communication scenarios, cumulative constraints (e.g., average power budgets or long-term fairness) exhibit a relatively ``soft'' temporal characteristic, where transient fluctuations are permissible provided the long-term expectation is satisfied. For such constraints, standard linear Lagrangian relaxation is often sufficient to guide the policy towards the feasible region without the numerical instability introduced by aggressive quadratic penalties. Consequently, to optimize the trade-off between training convergence speed and strict safety enforcement, we adopt a hybrid penalty formulation in the subsequent theoretical analysis and experimental validation: we employ pure linear dual ascent for cumulative constraints to ensure stable gradient estimation, while reserving the augmented quadratic penalty terms strictly for instantaneous hard constraints. This strategic decomposition prevents the reward landscape from becoming overly distorted, thereby significantly accelerating the learning process while guaranteeing rigorous satisfaction of critical safety bounds.

The augmented Lagrangian for Problem~\eqref{eq:safe_rl_objective} addresses different constraint types through specialized penalty mechanisms:
\begin{align}
&\mathcal{L}^{aug}(\boldsymbol{\lambda}, \boldsymbol{\mu}, \boldsymbol{\nu}, \pi) \notag \\ = & V_r(\pi) - \sum_{i \in \mathcal{I}} \Lambda_i - \sum_{j \in \mathcal{J}} \Upsilon_j - \sum_{k \in \mathcal{K}} \Psi_k  
\end{align}
where the penalty terms are defined as:
\begin{align}
\Lambda_i &= \lambda_i ( \hat{V}_{c_i}(\pi) - d_i ), \\
\Upsilon_j &= \mathbb{E}_\pi \left[ \sum_{t=0}^{\infty} \gamma^t \underbrace{ \left( \mu_j e_j(s_t,a_t) + \frac{\rho_j}{2} e_j(s_t,a_t)^2 \right)}_{\eqqcolon \phi_{\text{eq},j}(s_t,a_t)} \right],  \label{lambdaPenalty}\\
\Psi_k &= \mathbb{E}_\pi \left[ \sum_{t=0}^{\infty} \gamma^t \underbrace{\left( \nu_k g_k^+ + \frac{\rho_k}{2} (g_k^+)^2 \right)}_{\eqqcolon \phi_{\text{inst},k}(s_t,a_t)} \right] \label{PsiPenalty}.
\end{align}
with $g_k^+ = \max(0, g_k(s_t, a_t))$. 

The dual variables are updated via projected dual ascent with the two-timescale structure:
\begin{align}
\label{dualUpadate1}
\lambda_i^{\tau+1} &= \left[ \lambda_i^\tau + \beta_\lambda \left( \hat{V}_{c_i}(\pi) - d_i \right) \right]_+, \quad \forall i \in \mathcal{I}
\end{align}
\begin{align}
\label{dualUpadate2}
\mu_j^{\tau+1} &= \mu_j^\tau + \beta_\mu \mathbb{E}_\pi[e_j(s_t,a_t)], \quad \forall j \in \mathcal{J}  
\end{align}
\begin{align}
\label{dualUpadate3}
\nu_k^{\tau+1} &= \left[ \nu_k^\tau + \beta_\nu \mathbb{E}_\pi[g_k^+] \right]_+, \quad \forall k \in \mathcal{K}
\end{align}
where $\beta_\lambda, \beta_\mu, \beta_\nu > 0$ are dual learning rates and $[x]_+ = \max(0, x)$. Penalty parameters are updated adaptively: 
$\rho_j^{\tau+1} = \xi \rho_j^\tau$, $\rho_k^{\tau+1} = \xi \rho_k^\tau$ where $\xi > 1$ ensures progressive constraint enforcement.

\subsection{Safe-Deep Q-Learning Algorithm}
\label{subsec:safe_q_learning_algorithm}

We now present Safe-Deep Q-Learning that operationalizes the augmented Lagrangian framework with specialized action selection for mixed-integer wireless decisions.

\subsubsection{Constraint-Aware Q-Function Updates}
The augmented Q-function integrates constraint penalties into temporal difference learning. Following the ALM framework, the infinite-horizon objective function becomes:
\begin{align}
R_{\boldsymbol{\rho}}(\pi,\boldsymbol{\lambda},\boldsymbol{\mu},\boldsymbol{\nu}) = \mathbb{E}_\pi \left[\sum_{t=0}^{\infty}\gamma^{t}(r(s_t,a_t) - \phi_{all}(s_t,a_t))\right],
\label{R_rho(pi, lambda)}
\end{align}
where $\phi_{all}(s_t,a_t) = \phi_\text{inst}(s_t,a_t) + \phi_{\text{eq}}(s_t,a_t) + \phi_{\text{cum}}$ combines all penalty terms, with 
\begin{align}
    \phi_\text{cum}& \coloneqq \sum_{i \in \mathcal{I}} \Lambda_i, \\ 
    \phi_\text{eq}(s_t,a_t) &\coloneqq \sum_{j \in \mathcal{J}} \phi_{\text{eq},j}(s_t,a_t), \\ 
    \phi_\text{inst}(s_t,a_t) &\coloneqq \sum_{k \in \mathcal{K}} \phi_{\text{inst},k}(s_t,a_t).
\end{align}
In these definitions, $\phi_{\text{inst}}$ and $\phi_{\text{eq}}$ represent the augmented Lagrangian penalties for instantaneous constraints, while $\phi_{\text{cum}}$ incorporates the cumulative budget requirements into the per-step reward signal. 
Conversely, $\phi_{\text{cum}}^{\text{step}}$ mathematically distributes the cumulative budget requirements into the per-step reward signal.
The augmented Q-function becomes:
\begin{align}
\label{ALQ-fun}
&Q^{aug}(s_t, a_t) \notag \\
= &r(s_t, a_t) + \gamma \max_{a'} Q^{aug}(s_{t+1}, a') - \phi_{step}(s_t,a_t)  
\end{align}
where $\phi_{step}(s,a) = \phi_\text{inst}(s,a) + \phi_\text{eq}(s,a) + \phi_\text{cum}^\text{step}(s,a)$.
Crucially, the step-wise cumulative penalty is formulated as:
\begin{align}
\phi_\text{cum}^\text{step}(s_t,a_t) = \sum_{i\in \mathcal{I}} \lambda_i \left( c_i(s_t,a_t) - (1-\gamma)d_i \right).
\end{align}
The subtraction of the term $(1-\gamma)d_i$ is mathematically imperative; it normalizes the long-term budget $d_i$ into an equivalent per-step allowance via the infinite geometric series of the discount factor $\gamma$. This ensures that the instantaneous penalized reward seamlessly aligns with the global trajectory constraints, thereby preventing biased gradient estimations and facilitating strictly stable dual variable convergence.
Parameters are updated by minimizing:
\begin{align}
\label{MSEQ}
\theta^{t+1} = \theta^t - \alpha_t \nabla_\theta \mathbb{E} \left[ \left( Q_\theta(s,a) - Q^{aug}(s,a) \right)^2 \right]  
\end{align}

\subsubsection{Mixed-Integer Safe Action Selection}

For wireless systems with general discrete variables $\mathbf{z}_t \in \mathcal{Z}$ (e.g., discrete beamforming indices or user-channel assignments) and continuous variables $\mathbf{y}_t$ (power allocation), we employ two-stage selection:

\textbf{Stage 1 - Discrete Selection:}
\begin{align}
\mathbf{z}_t = \begin{cases}
\arg\max_{\mathbf{z} \in \mathcal{Z}_{safe}(s_t)} Q_\theta(s_t, (\mathbf{z}, \mathbf{y}^*(\mathbf{z}))) & \text{w.p. } 1-\epsilon \\
\text{uniform from } \mathcal{Z}_{safe}(s_t) & \text{w.p. } \epsilon
\end{cases} 
\end{align}
where $\mathcal{Z}_{safe}(s_t) = \{\mathbf{z} \in \mathcal{Z} : \exists \mathbf{y}, g_k(s_t, (\mathbf{z}, \mathbf{y})) \leq 0, \forall k\}$ denotes the \textit{state-dependent feasible discrete action set}. 
It filters out discrete decisions $\mathbf{z}$ that cannot satisfy the instantaneous hard constraints $g_k \leq 0$ for any continuous $\mathbf{y}$, thereby pruning the search space and ensuring that the agent's exploration is restricted to physically attainable regions.

\textbf{Stage 2 - Continuous Optimization:}
\begin{align}
\mathbf{y}_t^* = \arg\max_{\mathbf{y} \in \mathcal{Y}(s_t)} Q_\theta(s_t, (\mathbf{z}_t, \mathbf{y})) \text{ s.t. } g_k(s_t, (\mathbf{z}_t, \mathbf{y})) \leq 0, \forall k.  
\end{align}

\begin{algorithm}[!t]
\caption{Safe-Deep Q-Learning for Optimizing Problem~\text{(P2)}}
\label{alg:safe_q_learning}
\SetAlgoLined
\SetKwInOut{Input}{Input}\SetKwInOut{Output}{Output}
\Input{State space $\mathcal{S}$, constraints $\{g_k\}, \{c_i\}$, thresholds $\{d_i\}$}
\Output{Safe policy $\pi^*$}
Initialize $Q_\theta$, $Q_{\hat{\theta}}$, $\mathcal{D} = \emptyset$, $\lambda_i^0 = 0$, $\mu_j^0 = 0$, $\nu_k^0 = 0$, $  \rho_j^0, \rho_k^0 > 0$\;
\For{episode $= 1, 2, \ldots$}{
    Observe $s_0$\;
    \For{$t = 0, 1, \ldots, T-1$}{
        Select $a_t$ using the local observations\;
        Execute $a_t$, observe $r_t$, $s_{t+1}$\;
        Compute $g_k(s_t, a_t)$, $e_j(s_t, a_t)$, $c_i(s_t, a_t)$\;
        Store $(s_t, a_t, r_t, s_{t+1}, \{g_k\}, \{e_j\}, \{c_i\})$ in $\mathcal{D}$\;
        \If{Q-update time}{
            Sample minibatch from $\mathcal{D}$\;
            Update $Q_\theta$ using Eq.~\eqref{ALQ-fun}\;
        }
    }
    Estimate $\hat{V}_{c_i}(\pi)$ and update $\boldsymbol{\lambda}, \boldsymbol{\nu}, \boldsymbol{\rho}$ using Eqs.~\eqref{dualUpadate1},~\eqref{dualUpadate2}, and \eqref{dualUpadate3}\;
}
\end{algorithm}

\begin{algorithm}[!t]
\caption{Multi-Agent Safe-Deep Q-Learning for Optimizing Problem~\text{(P2)}}
\label{alg:multi_agent_safe_q_learning}
\SetAlgoLined
\SetKwInOut{Input}{Input}\SetKwInOut{Output}{Output}
\Input{Agent set $\mathcal{N}$, state spaces $\{\mathcal{S}^n\}$, action spaces $\{\mathcal{A}^n\}$, constraints $\{g_k\}, \{e_j\}, \{c_i\}$, thresholds $\{d_i\}$}
\Output{Safe policies $\{\pi_n^*\}$ for all agents $n \in \mathcal{N}$}
\For{each agent $n \in \mathcal{N}$}{
    Initialize $Q_{\theta^n}$, $Q_{\hat{\theta}^n}$, Replay Buffer $\mathcal{D}^n$\;
    Initialize local dual variables $\lambda_i^n = 0$, $\mu_j^n = 0$, $\nu_k^n = 0$, penalty factors  $\rho_j^n > 0$, $\rho_k^n > 0$\;
}
\For{episode $= 1, 2, \ldots, M$}{
    Reset environment, observe initial states $s_0^n$\;
    \For{$t = 0, 1, \ldots, T-1$}{
        \For{each agent $n \in \mathcal{N}$}{
            Select action $a_t^n$ using $\epsilon$-greedy policy on $Q_{\theta^n}(s_t^n, \cdot)$\;
        }
        Execute joint action $\mathbf{a}_t = \{a_t^1, \dots, a_t^N\}$\;
        Observe rewards $r_t^n$, next states $s_{t+1}^n$, and constraint costs $\mathbf{c}_t$, $\mathbf{g}_t$, $\mathbf{e}_t$\;
        
        \For{each agent $n \in \mathcal{N}$}{
            Store transition $(s_t^n, a_t^n, r_t^n, \mathbf{c}_t, \mathbf{g}_t, \mathbf{e}_t, s_{t+1}^n)$ in $\mathcal{D}^n$\;
            \If{update time}{
                Sample minibatch $B$ from $\mathcal{D}^n$\;
                Compute augmented Q-functions $Q^{aug}_n(s_t, a_t)$ by Eq~\eqref{ALQ-fun}; \\
                Update $Q_{\theta^n}$ by minimizing Loss Eq.~\eqref{MSEQ}\;
            }
        }
    }
    \For{each agent $n \in \mathcal{N}$}{
        \If{constraint violated}{
            Update penalty parameter 
$\boldsymbol{\lambda}^n$, $\boldsymbol{\mu}^n$, $\boldsymbol{\nu}^n$, and penalty factors ${\rho}_i^n$, ${\rho}_j^n$, ${\rho}_k^n$\;
        }
    }
}
\end{algorithm}

The algorithm enforces a strict two-timescale learning hierarchy to guarantee convergence stability.
The fast timescale governs the primal policy improvement, where Q-network parameters $\theta$ are updated at the step-level frequency within the inner loop (Lines 9--11 of \textbf{Algorithm~\ref{alg:safe_q_learning}}). 
This allows the agent to rapidly adapt to the current penalty landscape.
Conversely, the slow timescale manages constraint enforcement, where dual variables $\boldsymbol{\lambda}$, $\boldsymbol{\mu}$, $\boldsymbol{\nu}$ are updated at the episode-level frequency (Line 14 of \textbf{Algorithm~\ref{alg:safe_q_learning}}) with a learning rate $\beta$ satisfying $\beta \ll \alpha$.
This frequency and step-size separation ensures that the dual variables remain quasi-static while the primal policy converges, thereby preventing primal-dual oscillations and ensuring asymptotic optimality.

\subsection{Training regime and safety shield mechanism}
\label{TrainingRegime}
To prevent system failures during the exploration phase, the Safe-Deep Q-Learning policy is trained offline using statistical channel models and simulated kinematics. 
A practical limitation of this approach is the simulation-to-reality gap. 
Physical environments introduce unmodeled dynamics, such as aerodynamic drag and transient radio frequency interference, which can cause a pre-trained agent to generate unsafe actions during deployment. 

A safety shield is integrated at the execution layer to handle these modeling discrepancies.
During operation, the agent outputs an intended action based on local observations. 
The safety shield evaluates this action prior to physical execution using real-time telemetry data, including inter-UAV distances obtained from DAA messages. 
If a constraint violation is predicted, the shield replaces the intended command with a predefined safe fallback maneuver, such as hovering and halting transmission. 
Although this override mechanism reduces the instantaneous communication throughput, it enforces collision avoidance.
Consequently, the system maintains physical safety upon deployment without requiring continuous online policy retraining.

\subsection{Extension to Multi-Agent Optimization}
\label{subsec:multi_agent_extension}
Wireless networks are inherently distributed systems comprising multiple interacting entities (e.g., UAV swarms, vehicular platoons). 
Directly applying centralized single-agent learning to such large-scale systems is often computationally intractable due to the exponential growth of the state-action space.
To address this scalability challenge, we extend the proposed framework to a multi-agent setting, as outlined in \textbf{Algorithm~\ref{alg:multi_agent_safe_q_learning}}.
Each UAV agent $n$ maintains its own local Q-network $Q_{\theta^n}$ and replay buffer $\mathcal{D}^n$.

The agents interact with the shared environment to collect local observations $s_t^n$ and select actions $a_t^n$ independently.
However, they share the global safety constraints (e.g., collision avoidance, minimum broadcast reliability). 
To enforce these, each agent incorporates the constraint violation costs directly into its local value learning process.

Specifically, the standard Q-learning target is modified by an augmented Lagrangian penalty term in Eqs.~\eqref{ALQ-fun}.
The penalty parameters are updated adaptively at the end of each episode based on whether the cumulative constraints were satisfied, effectively creating a ``safety price” that guides all independent agents towards a globally feasible equilibrium without requiring a central controller during execution.

\section{Convergence Analysis}
\label{sec:theoretical_analysis}
In this section, we provide the theoretical convergence analysis. 
We analyze the underlying constrained optimization problem (P2) to demonstrate that the proposed dual-timescale learning framework converges to an optimal policy that satisfies both instantaneous hard constraints and cumulative soft constraints, under mild assumptions. 

\subsection{Preliminaries and Assumptions}
We first analyze the underlying constrained optimization problem. To facilitate the convergence analysis, we consider a reformulated version of (P2), denoted as (P2'), and let $P^*$ denote its optimal objective value:
\begin{align*}
\text{(P2'):}~~ &P^*=\max_{\pi} ~ V_r(\pi) = \mathbb{E}_{\pi}\left[ \sum_{t=0}^{\infty} \gamma^t r(s_t, a_t) \right] \\
\text{s.t.} \quad  & V_{c_i}(\pi) = \mathbb{E}_{\pi}\left[ \sum_{t=0}^{\infty} \gamma^t \tilde{c}_i(s_t, a_t) \right] \leq 0, ~\forall i \in \mathcal{I},  \\
&g_k(s_t, a_t) \leq 0, \quad \forall k \in \mathcal{K},~\forall t \geq 0,
\end{align*}
Here, (P2') does not differ from (P2) much, except that it drops the equality instantaneous constraints. This is because an equality can be induced by two inequalities. 

We adopt the following standard assumptions~\cite{9683088,10.5555/3454287.3454966}:
\begin{assumption}[cf.~{\cite[Assumption~1]{9683088}}]
Suppose the feasible policy set for (P2) has a non-empty relative interior. Namely, (P2) satisfies Slater's condition. Furthermore, we suppose that any policy $\pi$ satisfying the aggregated constraint $\mathbb{E}_\pi[\sum \gamma^t g_k(s_t, a_t)] \leq 0$ also satisfes the per-step constraint $\mathbb{E}[g_k(s_t, a_t)] \leq 0, \forall t \in \{0,1,... \}$. 
\label{ass:slater}
\end{assumption}

\begin{assumption}[Boundedness]
\label{ass:bounded}
The reward $r$, cumulative cost functions $\tilde{c}_i$, and constraint functions $g_k$ are bounded by a constant $B < \infty$.
\end{assumption}

\begin{proposition}
    Under Assumption~\ref{ass:slater}, strong duality holds for (P2'). That is, 
    \begin{equation*}
        P^* = \min_{\boldsymbol{\lambda} \geq 0,\boldsymbol{\nu} \geq 0} d(\boldsymbol{\lambda},\boldsymbol{\nu}).
    \end{equation*}
    Here, 
    \begin{align}
    d(\boldsymbol{\lambda},\boldsymbol{\nu}) \coloneqq 
    \max_{\pi} \mathbb{E}_{\pi} \left[ \sum_{t=0}^\infty \gamma^t \left( r - \sum_{i \in \mathcal{I}}\lambda_i\tilde{c}_i - \sum_{k\in \mathcal{K}} \nu_k g_k \right) \right]  \notag 
\end{align}
denotes the Lagrangian of (P2'). 
\end{proposition}

\begin{remark}[Mathematical and Physical Justification]
Assumption~\ref{ass:slater} is fundamental to the convergence analysis as it establishes \textit{Strong Duality} for the non-convex CMDP. Physically, Slater's condition implies the existence of a strictly feasible ``fallback” policy (e.g., minimum power transmission or UAV hovering) that respects all safety boundaries.
Crucially, by postulating that cumulative feasibility implies instantaneous feasibility in expectation~\cite{9683088}, this assumption ensures the duality gap vanishes, thereby validating the use of the augmented Lagrangian method for solving Problem \text{(P2)}. 
Assumption~\ref{ass:bounded} is inherently satisfied in wireless systems due to strict hardware limitations (e.g., peak transmit power $P_{\max}$) and the finite dynamic range of channel fading, ensuring that the dual variables and penalty terms in the proof remain bounded.
\end{remark}

\begin{remark}
    Assumption~\ref{ass:slater} does look stringent. However, this could hold, as one can use the "clipping" method, or to restrict the feasible policy class, both proposed in \cite[Proposition 1]{9683088}. This work \cite{9683088} also provides a reasonable endorsement for Assumption~\ref{ass:slater}.
\end{remark}

\subsection{Augmented Lagrangian and Zero Duality Gap}
Recall that $d(\boldsymbol{\lambda}, \boldsymbol{\nu})$ is the standard dual function. Now we define $d_\rho(\boldsymbol{\lambda}, \boldsymbol{\nu})$ as the augmented dual function. That is, 
\begin{align}
    &d_{\rho}(\boldsymbol{\lambda},\boldsymbol{\nu}) \coloneqq  \notag \\ &
    \max_{\pi} \mathbb{E}_{\pi} \left[ \sum_{t=0}^\infty \gamma^t \left( r - \sum_{i \in \mathcal{I}}\lambda_i\tilde{c}_i - \sum_{k\in \mathcal{K}}( \nu_k g_k^+ +\frac{\rho}{2} (g_k^+)^2 )  \right) \right]. \notag
\end{align} 

Based on Assumption \ref{ass:slater}, we have the following lemma (proven via standard ALM theory):
\begin{lemma}[Zero Duality Gap]
Let $(\boldsymbol{\lambda}^*, \boldsymbol{\nu}^*)$ be the optimal dual variables of the original problem. For any $\rho \geq 0$, let $(\boldsymbol{\lambda}^*_\rho, \boldsymbol{\nu}^*_\rho) = \arg \min d_\rho(\boldsymbol{\lambda}, \boldsymbol{\nu})$. Then
\begin{equation*}
    d_\rho(\boldsymbol{\lambda}^*_\rho, \boldsymbol{\nu}^*_\rho) = d(\boldsymbol{\lambda}^*, \boldsymbol{\nu}^*) = P^*.
\end{equation*}
Here, $(\boldsymbol{\lambda}^*_\rho, \boldsymbol{\nu}^*_\rho) = \arg \min d_{\rho}(\boldsymbol{\lambda}, \boldsymbol{\nu})$, and $(\boldsymbol{\lambda}^*, \boldsymbol{\nu}^*) = \arg \min d(\boldsymbol{\lambda}, \boldsymbol{\nu}).$
\end{lemma}

\subsection{Convergence of Instantaneous Constraints}
We analyze the behavior of the optimal policy $\pi_\rho^*$ for the augmented problem as the penalty parameter $\rho \to \infty$.

\begin{theorem}[Asymptotic Feasibility]
\label{thm:instantaneous}
As $\rho \to \infty$, the optimal policy of the augmented problem satisfies the instantaneous constraints. Specifically:
\begin{equation*}
    \lim_{\rho \to \infty} \Psi(\pi_{\rho}^*)= \mathbb{E}_{\pi_\rho^*} \left[ \sum_{t=0}^\infty \gamma^t \cdot (g_k^+(s_t, a_t))^2 \right] = 0.
\end{equation*}
Here, $\pi_\rho^* = \arg \max_{\pi} R_\rho(\pi, \boldsymbol{\lambda}^*_\rho, \boldsymbol{\nu}^*_\rho)$. 

\end{theorem}

\begin{IEEEproof}
From Lemma 1, we know $d_\rho(\boldsymbol{\lambda}^*_\rho, \boldsymbol{\nu}^*_\rho) = P^*, \forall \rho \geq 0$. Expanding the augmented dual function:
\begin{align*}
    P^* &= V_0(\pi_\rho^*) - \boldsymbol{\lambda}^*_\rho V_c(\pi_\rho^*) - \boldsymbol{\nu}^*_\rho G(\pi_\rho^*) - \frac{\rho}{2}\Psi(\pi_\rho^*), 
\end{align*}
where $G_k(\pi) = \mathbb{E}_\pi[\sum_{t=0}^{\infty} \gamma^{t} g_k^+(s_t, a_t) ]$.
We consider the magnitude of the augmented dual variable $\boldsymbol{\nu}^*_\rho$ relative to the original optimal dual $\boldsymbol{\nu}^*$:

\textbf{Case 1:} $\nu^*_{\rho, k} > \nu^*_k$.
By the properties of the dual function and the quadratic penalty term (which is non-negative), using a larger dual variable in the surrogate function would imply $d_\rho(\boldsymbol{\lambda}^*_\rho, \boldsymbol{\nu}^*_\rho) < d(\boldsymbol{\lambda}^*, \boldsymbol{\nu}^*) = P^*$, which contradicts Lemma 1.

\textbf{Case 2:} $\nu^*_{\rho, k} \leq \nu^*_k$.
In this case, $\|\boldsymbol{\nu}^*_\rho\|$ is bounded by $\|\boldsymbol{\nu}^*\|$. Using the boundedness of $V_0, V_c, G$ (Assumption \ref{ass:bounded}, bounded by $B$):
\begin{align*}
    P^* &\leq B + \|\boldsymbol{\lambda}^*_\rho\| B + \|\boldsymbol{\nu}^*\| B - \frac{\rho}{2}\delta.
\end{align*}
Rearranging the terms, we get:
\begin{equation*}
    \frac{\rho}{2}\delta \leq B(1 + \|\boldsymbol{\lambda}^*_\rho\| + \|\boldsymbol{\nu}^*\|) - P^*.
\end{equation*}
As $\rho \to \infty$, the LHS approaches $\infty$ while the RHS remains bounded (since dual variables of the cumulative constraints $\boldsymbol{\lambda}^*_\rho$ stabilize for sufficient $\rho$). This leads to a contradiction.
Thus, $\lim_{\rho \to \infty} \Psi(\pi_\rho^*) = 0$.
\end{IEEEproof}

\subsection{Convergence of Cumulative Constraints}
We now address the cumulative constraints handled by the primal-dual updates on $\boldsymbol{\lambda}$.

\begin{theorem}[Cumulative Constraint Satisfaction]
\label{thm:Cumulative}
Let $\{\pi^{(s)}\}_s$ be the sequence of policies generated by any primal-dual algorithm. Then there exists a subsequence converging to a policy feasible to cumulative constraints, i.e.,
\begin{equation*}
    \liminf_{s \to \infty} V_{c_i}(\pi^{(s)}) \leq 0, \quad \forall i \in \mathcal{I}.
\end{equation*}
\end{theorem}

\begin{IEEEproof}
Assume the contrary: $\liminf_{s \to \infty} V_{c_i}(\pi^{(s)}) \geq \epsilon > 0$ for some constraint $i$.
The dual update rule is $\lambda_i^{(s+1)} = [\lambda_i^{(s)} + \beta_s V_{c_i}(\pi^{(s)})]_+$.
If the violation consistently exceeds $\epsilon$, the dual variable diverges: $\lambda_i^{(s)} \to \infty$ (since $\beta_s$ can be chosen to satisfy Robbins–Monro condition, thus $\sum \beta_s = \infty$).
Consider the objective function maximized by $\pi^{(s)}$:
\begin{equation*}
    \max_\pi \left( V_0(\pi) - \lambda_i^{(s)} V_{c_i}(\pi) - \text{other penalty terms} \right).
\end{equation*}
As $\lambda_i^{(s)} \to \infty$, the term $-\lambda_i^{(s)} V_{c_i}(\pi)$ dominates. Any policy with $V_{c_i}(\pi) > 0$ will yield an objective value approaching $-\infty$.
However, we know that a feasible policy exists (Slater's condition) which gives a finite objective value. Thus, the optimization step must eventually select a policy that satisfies the constraint to avoid the infinite penalty, contradicting the assumption that $V_{c_i} > \epsilon$ for all $s$.
\end{IEEEproof}

\begin{remark}[Holistic Convergence and Architectural Implications]
The combination of Theorem~\ref{thm:instantaneous} and Theorem~\ref{thm:Cumulative} provides a complete theoretical guarantee for the proposed Safe-Deep Q-Learning framework. 
Theorem~\ref{thm:instantaneous} confirms that the quadratic penality coefficients asymptotically force the policy into the safe region regarding hard constraints (e.g., collision avoidance), while Theorem~\ref{thm:Cumulative} ensures that the dual ascent stabilizes the long-term resource budgeting (e.g., energy limits). 
This validates our dual-timescale design rationale: by mathematically decoupling the handling of immediate physical safety from cumulative resource management, the algorithm eliminates the primal-dual oscillations common in constrained RL, ensuring convergence to a policy that is both physically safe and asymptotically optimal.
\end{remark}

\begin{figure}[!tbp]
\centerline{\includegraphics[width=.475\textwidth]{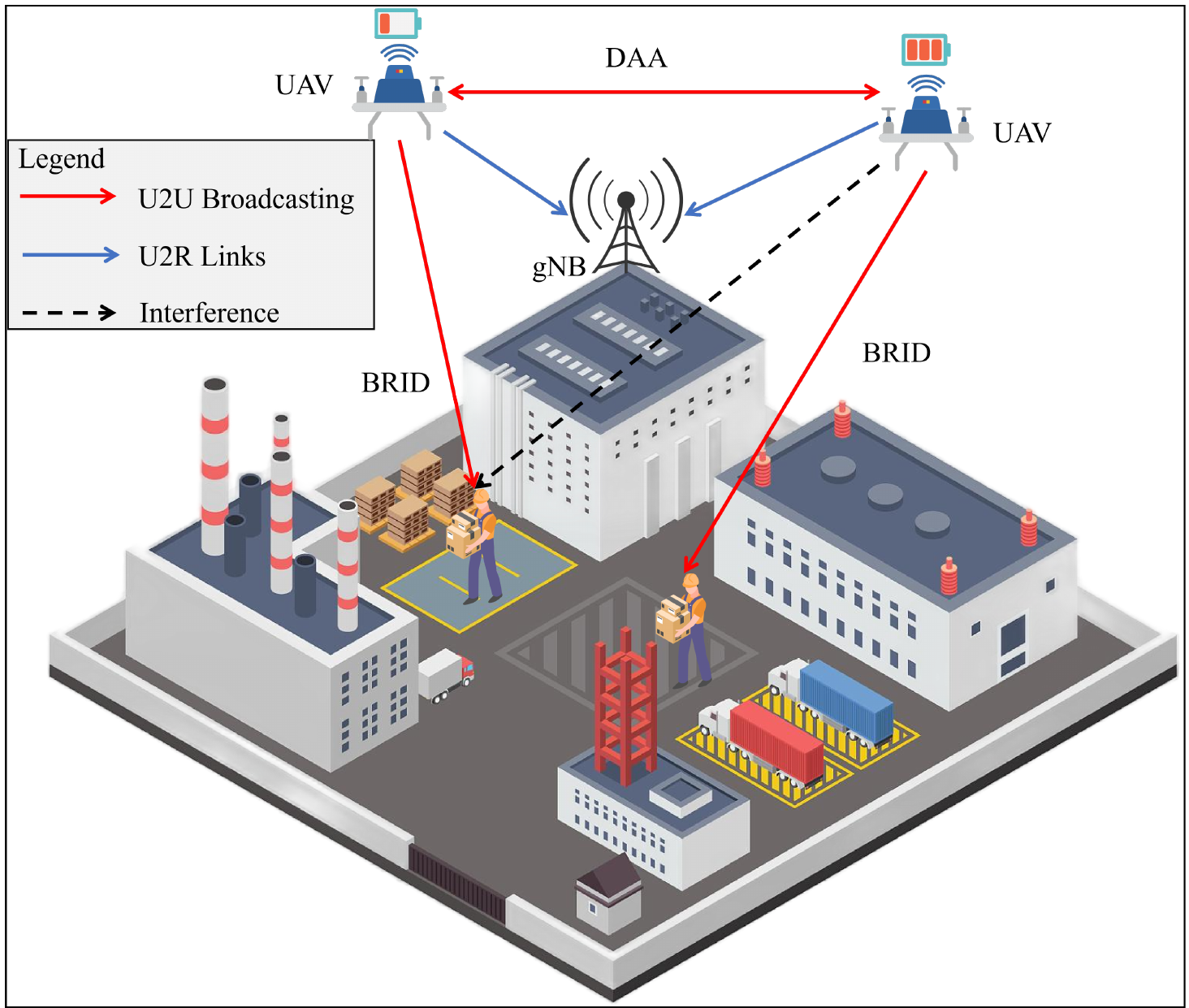}}
\caption{
Multi-UAV mobility coverage system.
}
\label{Fig:System_model1}
\vspace{-10pt}
\end{figure}
\section{Study case 1: Dynamic service provisioning with safety constraints}
\label{subsec:uav_system}
As illustrated in Fig.~\ref{Fig:System_model1}, we consider a mission-critical multi-UAV network where $N_u$ UAVs function as aerial base stations.
These agents must simultaneously support reliable UAV-to-UAV (U2U) coordination and high-bandwidth UAV-to-Ground (U2R) links.

\subsection{System Model}
\label{UAV_SYS_MODEL}
Each UAV broadcasts detect-and-avoid (DAA) messages to other UAVs and ground users via URLLC links while maintaining eMBB connections to the ground base station (gNB). Resource blocks are allocated across $B$ orthogonal subchannels in 1ms mini-slots, with UAVs independently selecting subchannels and power levels based on channel states.
To preclude resource contention between critical control signaling and the data plane under optimization, the kinematic states of the agents (via DAA messages) are disseminated over a strictly orthogonal, pre-allocated URLLC control channel. This out-of-band signaling architecture eliminates circular dependencies, guaranteeing the deterministic, ultra-low latency spatial awareness prerequisite for strict collision avoidance.

The channel between UAV $n_u$ and receiver $n_r$ (UAV or ground user) over subchannel $b$ follows Rician fading:
\begin{align}
h_{n_u,n_r}[b] = \nu_{n_u,n_r}f_{n_u,n_r}[b],
\end{align}
where $\nu_{n_u,n_r}$ captures path loss and shadowing and $f_{n_u,n_r}[b]$ represents small-scale fading.

\subsection{UAV Movement Model}
As shown in Fig.~\ref{fig:distance}, each UAV possesses the capability to dynamically adjust its trajectory to optimize service coverage performance while ensuring collision avoidance with other UAVs. As illustrated in Fig.~\ref{fig:movement}, the movement model employs a discrete action space comprising five mobility primitives: four directional movements (forward, backward, left, right) with a fixed step size of one meter, and one hovering action (remaining stationary). 
The position update is:
\begin{align}
\mathbf{p}_{n_u}(t+1) = \mathbf{p}_{n_u}(t) + \Delta\mathbf{p}_{n_u}(a_{n_u}(t))
\end{align}
where $\Delta\mathbf{p}_{n_u} \in \{(1,0,0), (-1,0,0), (0,1,0), (0,-1,0), (0,0,0)\}$. Collision avoidance is enforced by $\|\mathbf{p}_{n_u}(t) - \mathbf{p}_{n_u'}(t)\|_2 \geq d_{\min}$, where ${n_u'}\neq{n_u}$ and ${n_u'} \in [1, N_u]$.

\begin{figure*}[!tbp]
\centerline{\includegraphics[width=.975\textwidth]{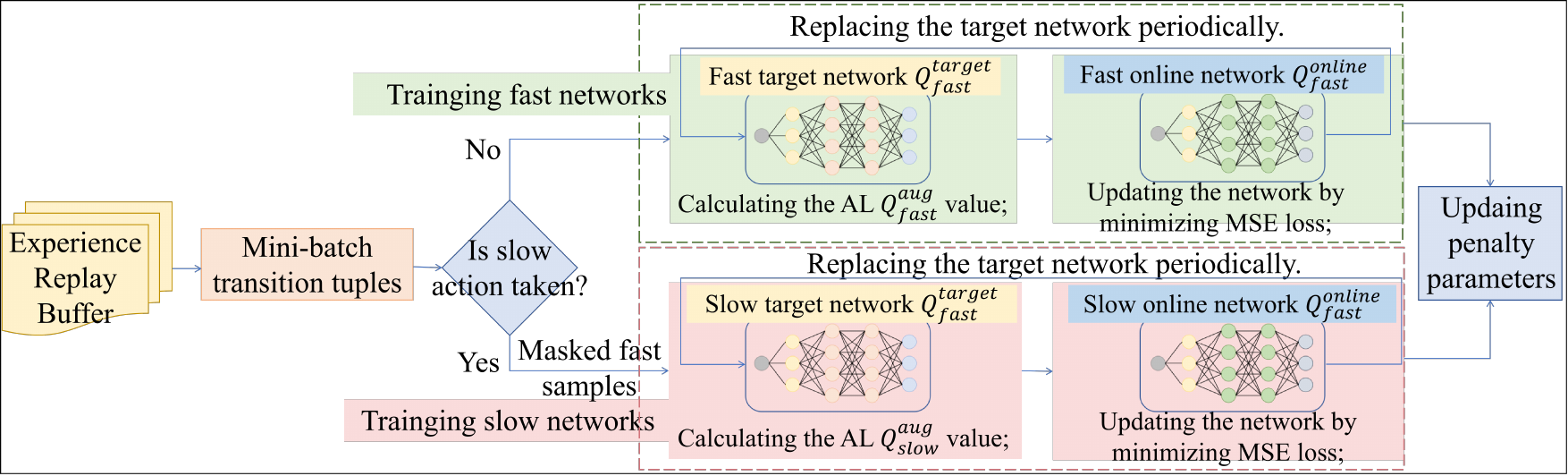}}
\caption{
The architecture of the proposed dual-timescale augmented Lagrangian Q-network. 
}
\label{Fig:Architecture}
\vspace{-10pt}
\end{figure*}

\begin{figure}[!tbp]
\centering
\subfigure[Inter-UAV safety distance constraint] {
\includegraphics[width=.225\textwidth]{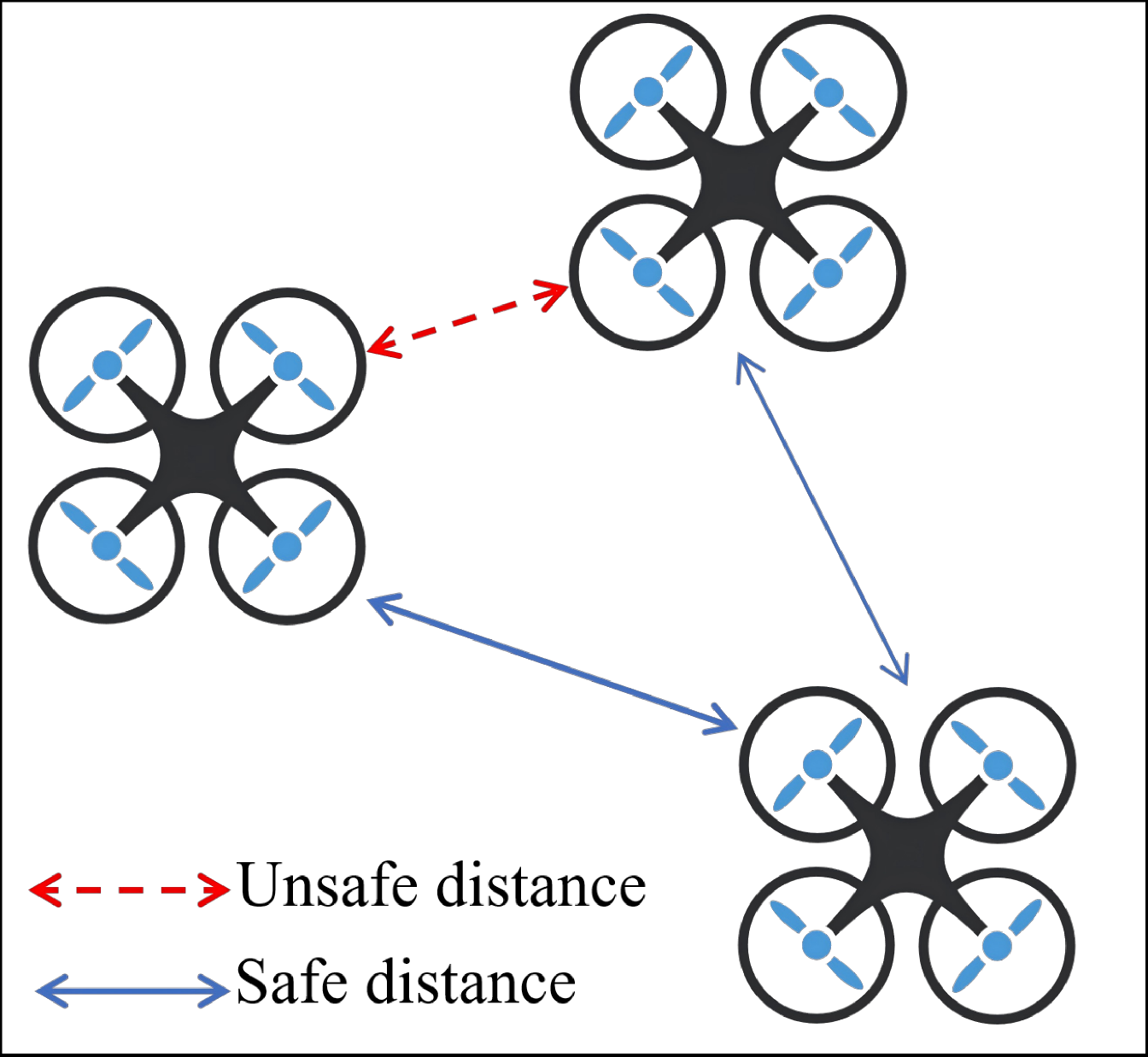}
\label{fig:distance}
}
\subfigure[UAV mobility action space] {
\includegraphics[width=.225\textwidth]{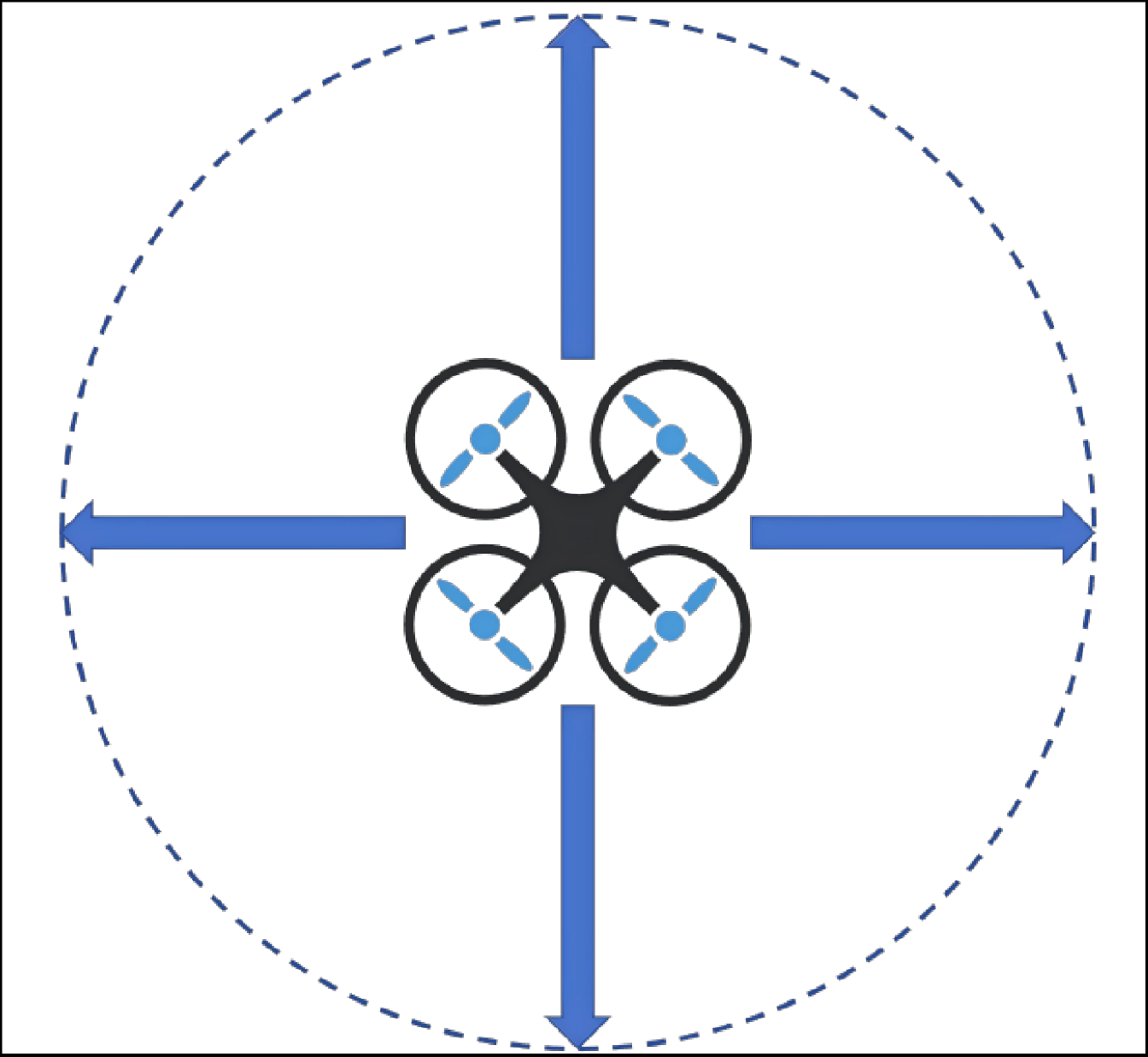}
\label{fig:movement}
}
\caption{
UAV operational constraints and mobility model: (a) Minimum inter-UAV distance requirement for collision avoidance, ensuring safe separation between aerial nodes; (b) Discrete action space for UAV mobility control, comprising five movement primitives (forward, backward, left, right, hover) for trajectory optimization while maintaining operational safety.
}
\label{fig:safe distance}
\end{figure}

\subsection{Transmission and Energy Models}
The SINR of the U2U link between UAV $n_u$ and receiver $n_r$ is:
\begin{align}
\text{SINR}_{n_u,n_r}[b] = \frac{{h}_{n_u,n_r}[b]P_{n_u}^{U}[b]\omega_{n_u}^{U}[b]}{\sigma^{2} + I_{n_u,n_r}[b]},
\end{align}
where $P_{n_u}^{U}[b]$ is transmit power, $\omega_{n_u}^{U}[b] \in \{0,1\}$ indicates subchannel usage, and $I_{n_u,n_r}[b]$ aggregates interference from other transmissions. The achievable rate is $C_{n_u,n_r}[b] = W_b \log_2(1 + \text{SINR}_{n_u,n_r}[b])$.

Similarly, U2R links to the gNB follow the same model with channel $\hat{h}_{n_u,R}[b]$ and rate $C_{n_u,R}[b]$. Each subchannel serves at most one transmission type: $\omega_{n_u}^{U}[b] + \omega_{n_u}^{R}[b] \leq 1$.

UAV energy consumption includes communication energy and operational overhead $\epsilon_o$:
\begin{align}
E_{n_u}(t) = E_{n_u}(t-1) - \left(\epsilon_o + \sum_{b=1}^{B}(P_{{n_u}}^{U}[b] + P_{{n_u}}^{R}[b])\Delta t\right).
\end{align}

\subsection{Constrained Optimization Formulation}

U2U reliability measures successful DAA message delivery:
\begin{align}
R_{n_u}^{\text{U2U}}(t) = \frac{1}{N_u}\sum_{n_u=1}^{N_u} \mathbb{I}\left\{\sum_{b=1}^{B} C_{n_u,n_r}[b]\omega_{n_u}^{U}[b]\Delta t \geq D\right\},
\end{align}
where $D$ is message size and $\mathbb{I}\{\cdot\}$ is the indicator function.

Total U2R throughput aggregates across all UAVs:
\begin{align}
R^{\text{U2R}}(t) = \sum_{n_u=1}^{N_u}\sum_{b=1}^{B} C_{n_u,R}[b]\omega_{n_u}^{R}[b].
\end{align}

The constrained optimization problem for each UAV $n_u$ aligns with the general CMDP framework:
\begin{align}
&\text{(P2):} \max_{\pi_{n_u}} \mathbb{E}\left[\sum_{t=1}^{T} \gamma^t R^{\text{U2R}}(t)\right] \label{eq:uav_objective} \\
&\text{s.t.}: R_{n_u}^{\text{U2U}}(t) \geq \eta_{\min}, \quad \forall t \in \{1,\dots,T\}, \label{eq:u2u_constraint} \\
&\qquad E_{n_u}(t) \geq E_{\min}, \quad \forall t \in \{1,\dots,T\}, \label{eq:energy_constraint} \\
&\qquad \omega_{n_u}^{U}[b] + \omega_{n_u}^{R}[b] \leq 1, \quad \forall b \in \{1,\dots,B\}, \\
&\qquad \|\mathbf{p}_{n_u}(t) - \mathbf{p}_{n_u'}(t)\|_2 \geq d_{\min},{n_u'}\neq{n_u}.
\label{eq:channel_constraint}
\end{align}
where $\pi_{n_u}$ is the policy governing power allocation, subchannel selection and trajectory optimization, $\eta_{\min}$ is the minimum U2U reliability threshold, and $E_{\min}$ is the minimum required energy for safe operation.

This UAV trajectory control and resource allocation problem exemplifies the three fundamental challenges addressed by Safe-Deep Q-Learning: mixed-integer optimization of discrete variables $\omega_{n_u}^{U}[b], \omega_{n_u}^{R}[b]$,  $P_{n_u}^{U}[b]$, and $P_{n_u}^{R}[b]$), stochastic dynamics (time-varying channels), and dual-timescale constraints (instantaneous U2U reliability and energy safety bounds).

\subsection{State-action representation}
Employing a dual-timescale hybrid architecture, the action space seamlessly integrates mobility control with communication resource allocation.
Slow-timescale decisions govern trajectory planning through five discrete movement primitives: forward, backward, left, right, and hover. 
Concurrently, fast-timescale actions manage spectrum sharing and power distribution across U2U broadcast and U2R unicast links. 
Each transmission action combines $\mathcal{Z}_p$ discrete power levels with $\mathcal{Z}_c$ orthogonal subchannels, yielding $\mathcal{Z}_p \times \mathcal{Z}_c$ options per link and $(\mathcal{Z}_p \times \mathcal{Z}_c)^2$ composite communication strategies. 
The complete action tuple thus coordinates locomotion and transmission in a unified decision framework.

The observation space constitutes a high-dimensional continuous vector providing comprehensive environmental characterization. It encapsulates time-varying channel propagation characteristics, incorporating both large-scale fading statistics and small-scale fast-fading components, for U2U and U2R communication links. Spatial configuration is represented through relative Euclidean distances to all cooperative agents, establishing crucial geometric constraints for collision prevention and formation control. This integrated state representation facilitates robust perception in dynamically evolving environments while maintaining operational safety under stringent latency requirements.

The relative Euclidean distances to all cooperative agents are required to evaluate spatial constraints. 
Practically, this decentralized spatial awareness is acquired via the periodic DAA telemetry broadcasts.
As detailed in Section \ref{UAV_SYS_MODEL}, these lightweight updates are disseminated over a dedicated orthogonal URLLC control channel to preclude data-plane contention.
To constrain the signaling overhead, the distance matrices are exclusively updated at the slow timescale, aligning the control plane latency with the physical mobility limits of the agents.

\subsection{Control Plane and CTDE Architecture}
The proposed multi-agent system operates under the Centralized Training with Decentralized Execution (CTDE) paradigm. 
During offline training, a centralized entity evaluates global collision risks and updates the dual variables $\boldsymbol{\lambda}$, $\boldsymbol{\mu}$, and $\boldsymbol{\nu}$. 
During execution, the control plane is fully decentralized. Each UAV $n$ relies solely on its local observation and the lightweight DAA messages broadcast by peers to make instantaneous routing and power control decisions without requiring a central coordinator.
Applying the methodology in Section \ref{TrainingRegime}, this study eschews the ab initio learning of fundamental UAV flight dynamics.
 Instead, we assume each UAV is equipped with a foundational kinematic controller that strictly enforces the minimum safety distance $d_{min}$ and emergency battery thresholds. 
The proposed Safe-Deep Q-Learning agent operates exclusively as a higher-order network orchestrator, synergizing this deterministic navigation baseline with dynamic spectrum and power allocation. Specifically, if the agent's intended spatial trajectory violates $d_{min}$ or exceeds the energy budget, the safety shield overrides the mobility command with a forced ``hover” state and curtails transmission power.
The ensuing constraint violation is then passed back to the Lagrangian dual variables, driving the multi-agent system toward a globally safe and network-optimal equilibrium.

\begin{figure}[!tbp]
\centerline{\includegraphics[width=.475\textwidth]{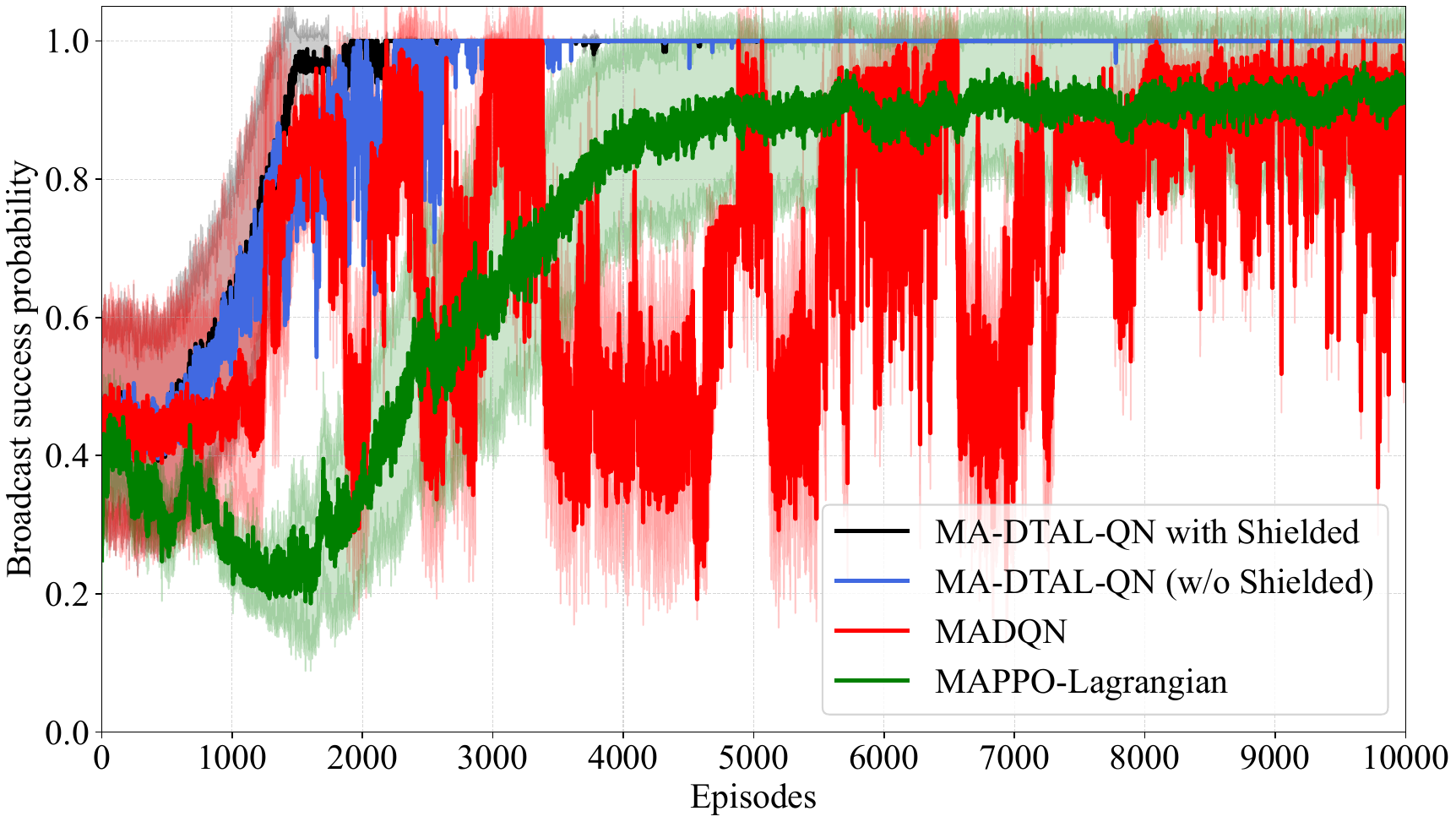}}
\caption{
Comparative Analysis of U2U Broadcast Constraint Convergence: Traditional MADQN exhibits persistent constraint violations throughout training, failing to maintain the required 100\% reliability threshold for safety-critical DAA messages. Our proposed MA-DTAL-QN achieves rapid constraint satisfaction within 4,000 episodes and maintains stable compliance with the reliability requirement, demonstrating the effectiveness of adaptive penalty scaling and constraint violation tracking mechanisms.
}
\label{fig:BroadcastProbability}
\vspace{-10pt}
\end{figure}

\begin{figure}[!tbp]
\centerline{\includegraphics[width=.475\textwidth]{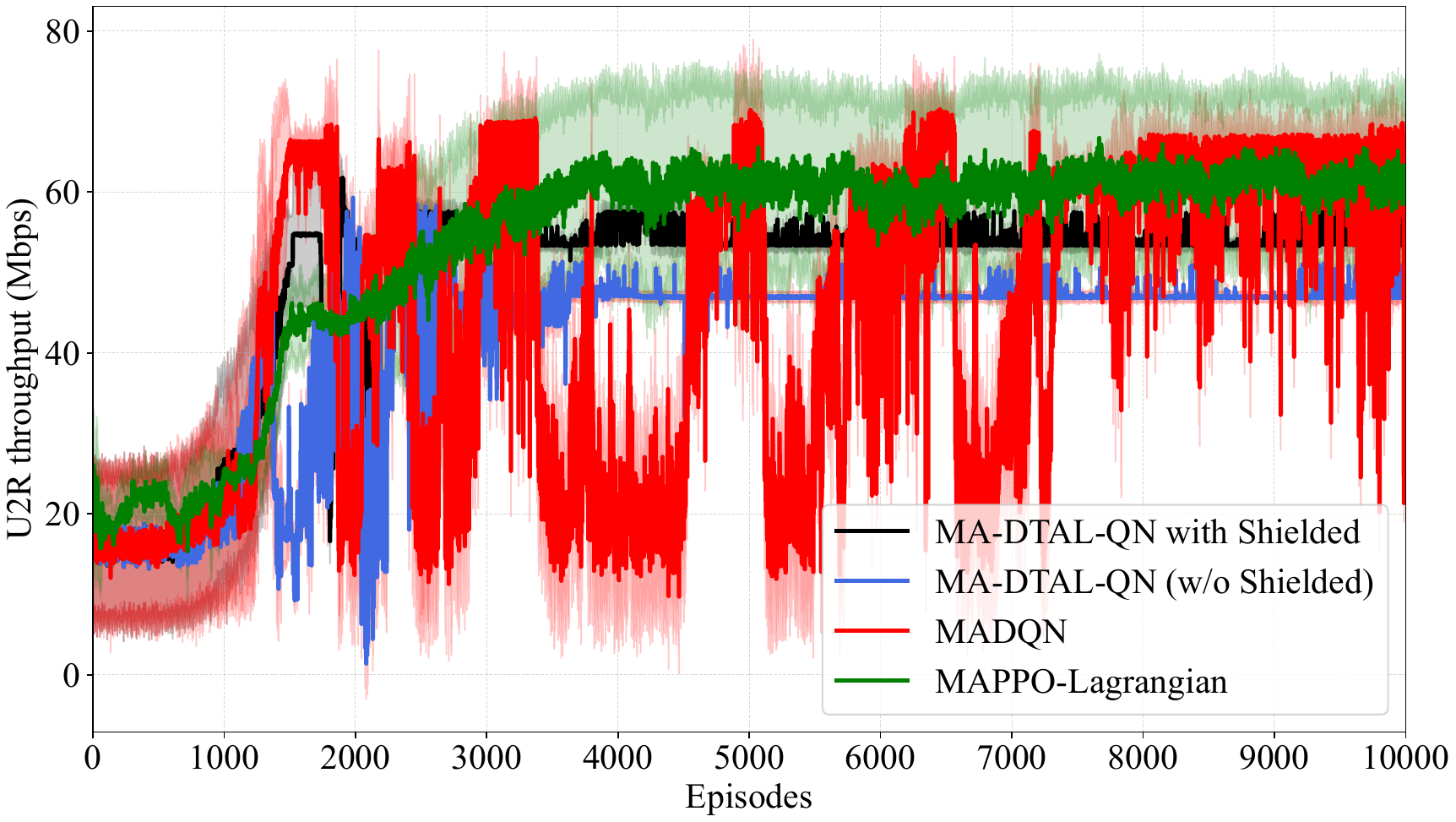}}
\caption{
U2R Throughput Optimization under Safety Constraints:  Conventional MADQN shows slow convergence with significant oscillations; Our MA-DTAL-QN achieves faster convergence with improved stability while maintaining all safety constraints.}
\label{fig:U2RThroughput}
\vspace{-10pt}
\end{figure}
\subsection{Evaluation results}
\begin{table}[htbp]
\renewcommand{\arraystretch}{1.2} %
\caption{Hyperparameters for Safe-DQL and Baselines}
\label{tab:hyperparameters}
\centering
\begin{tabular}{p{0.55\columnwidth} | p{0.3\columnwidth}}
\hline
\hline
\textbf{Parameter \& Symbol} & \textbf{Value} \\ 
\hline
\multicolumn{2}{c}{\textbf{Network Architecture}} \\ 
\hline
Hidden Layers (Safe-DQL \& MAPPO) & 2 $\times$ 128 units \\ 
Activation Function & ReLU \\
Optimizer & Adam \\ 
\hline
\multicolumn{2}{c}{\textbf{RL Training Process}} \\ 
\hline
Replay Buffer Capacity $\mathcal{D}$ & $50,000$ \\ 
Mini-Batch Size $|\mathcal{B}|$ & $1024$ \\ 
Discount Factor $\gamma$ & $0.95$ \\ 
Target Replace Freq. (Safe-DQL) & $100$ steps \\
Exploration $\epsilon$ (Safe-DQL) & Linear decay $[1, 0]$ \\
PPO Clip Ratio $\epsilon_{\text{clip}}$ & $0.2$ \\ 
PPO Update Epochs $K_{\text{epochs}}$ & $4$ \\ 
\hline
\multicolumn{2}{c}{\textbf{Learning Rates (Primal \& Dual)}} \\ 
\hline
Primal LR $\alpha$ (Safe-DQL) & $2 \times 10^{-5}$ \\
Primal LR (MAPPO Actor/Critic) & $3 \times 10^{-4}$ \\
Dual LRs $\beta_\lambda, \beta_\mu, \beta_\nu$ & $0.1$ \\ 
\hline
\multicolumn{2}{c}{\textbf{Augmented Lagrangian Parameters}} \\ 
\hline
Initial Penalty Multipliers $\boldsymbol{\rho}_0$ & $[0.05, 0.05, 0.05]$ \\ 
Penalty Scaling Factor $\xi$ ($c_{\rho}$) & $1.1$ \\ 
Max Penalty Bound $\rho_{\text{max}}$ & $100,000$ \\ 
\hline
\multicolumn{2}{c}{\textbf{Environment \& Dual-Timescale Dynamics}} \\ 
\hline
Fast-Timescale Step $\Delta t$ & $1$ ms \\ 
Slow-Timescale Interval & $20$ steps \\
Max Steps per Episode $T$ & $100$ \\
\hline
\hline
\end{tabular}
\end{table}
\subsubsection{Training performance evaluation}
To validate the effectiveness of Safe-Deep Q-Learning in safety-critical wireless applications, we conducted comprehensive experiments on the multi-UAV system described in Section IV-A. Our evaluation focuses on four key performance metrics: U2U broadcast reliability, U2R throughput, collision avoidance, and energy efficiency.
The hyperparameter configurations for the proposed Safe-Deep Q-Learning and Baseline Algorithms are summarized in Table~\ref{tab:hyperparameters}.

Figure \ref{fig:BroadcastProbability} evaluates the U2U broadcast reliability. The unconstrained MADQN predictably fails to satisfy the safety threshold. While the MAPPO-Lagrangian baseline shows improvement, it plateaus below the required 100\% reliability, reflecting the gradient variance issues inherent in on-policy penalty updates. 
Conversely, both variants of the proposed MA-DTAL-QN successfully converge to full compliance.
Notably, the shielded MA-DTAL-QN achieves strict stabilization by episode 2,000, outpacing its unshielded counterpart, demonstrating that the deterministic physical shield concurrently accelerates the learning of soft networking constraints by truncating unsafe exploration.

Figure \ref{fig:U2RThroughput} plots the corresponding U2R throughput. The MADQN yields nominal peaks near 70 Mbps, but these values are physically invalid due to continuous constraint violations. Similarly, the 60 Mbps capacity achieved by MAPPO-Lagrangian remains an artifact of its incomplete safety compliance. The proposed shielded MA-DTAL-QN converges to a highly stable plateau of 55 Mbps. Interestingly, this outperforms the unshielded variant, which stabilizes near 48 Mbps. This phenomenon validates that the execution-level shield not only guarantees absolute physical safety but also actively guides the policy search away from detrimental states, ultimately yielding a superior and strictly feasible communication performance.

\begin{figure}[!tbp]
\centerline{\includegraphics[width=.475\textwidth]{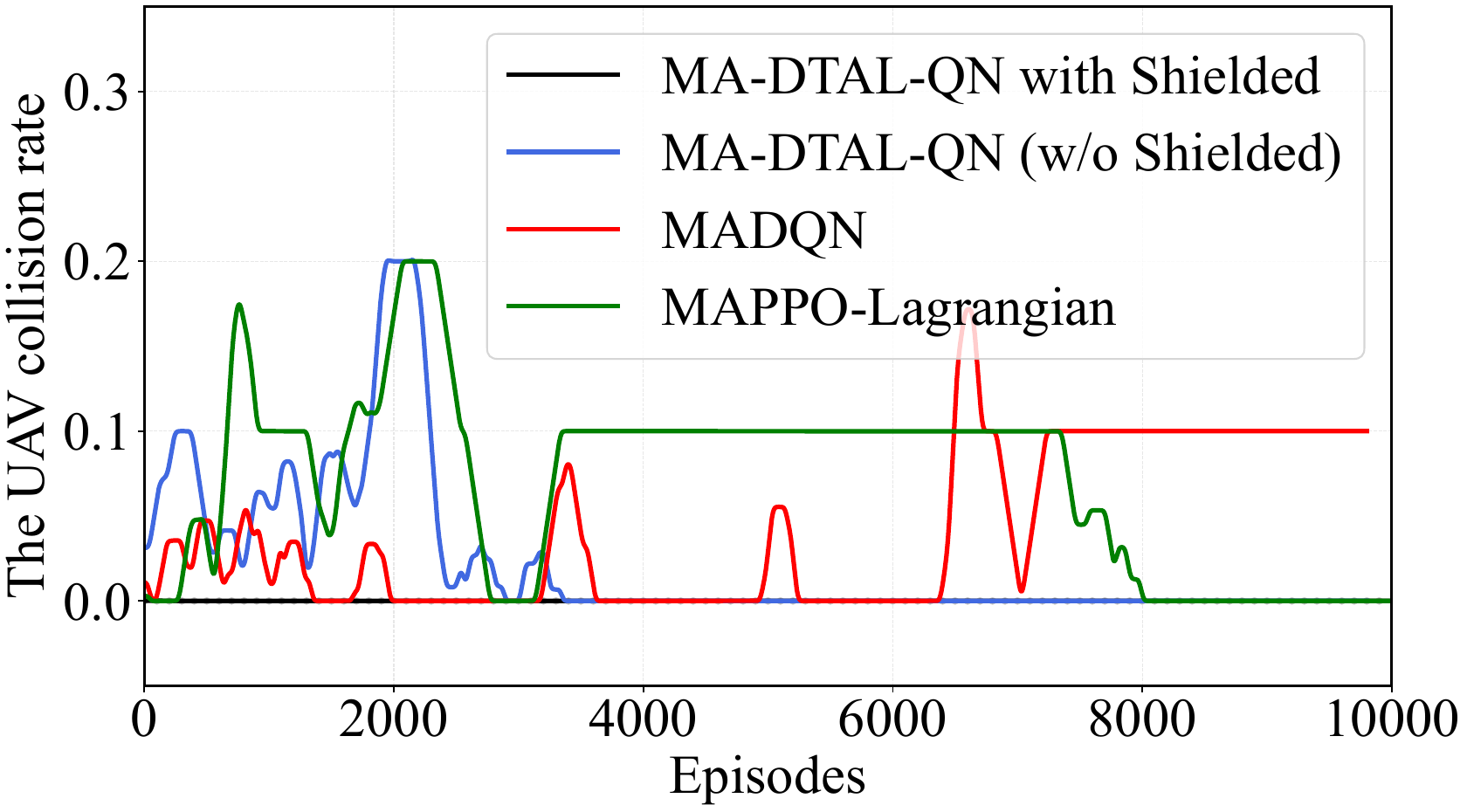}}
\caption{
Collision Avoidance Performance: Safe-Deep Q-Learning achieves zero collision violations after convergence, successfully maintaining the 30-meter safety distance between all UAV pairs through barrier functions and safe action projection.}
\label{fig:collision_avoidance}
\vspace{-10pt}
\end{figure}

\begin{figure}[!tbp]
\centerline{\includegraphics[width=.475\textwidth]{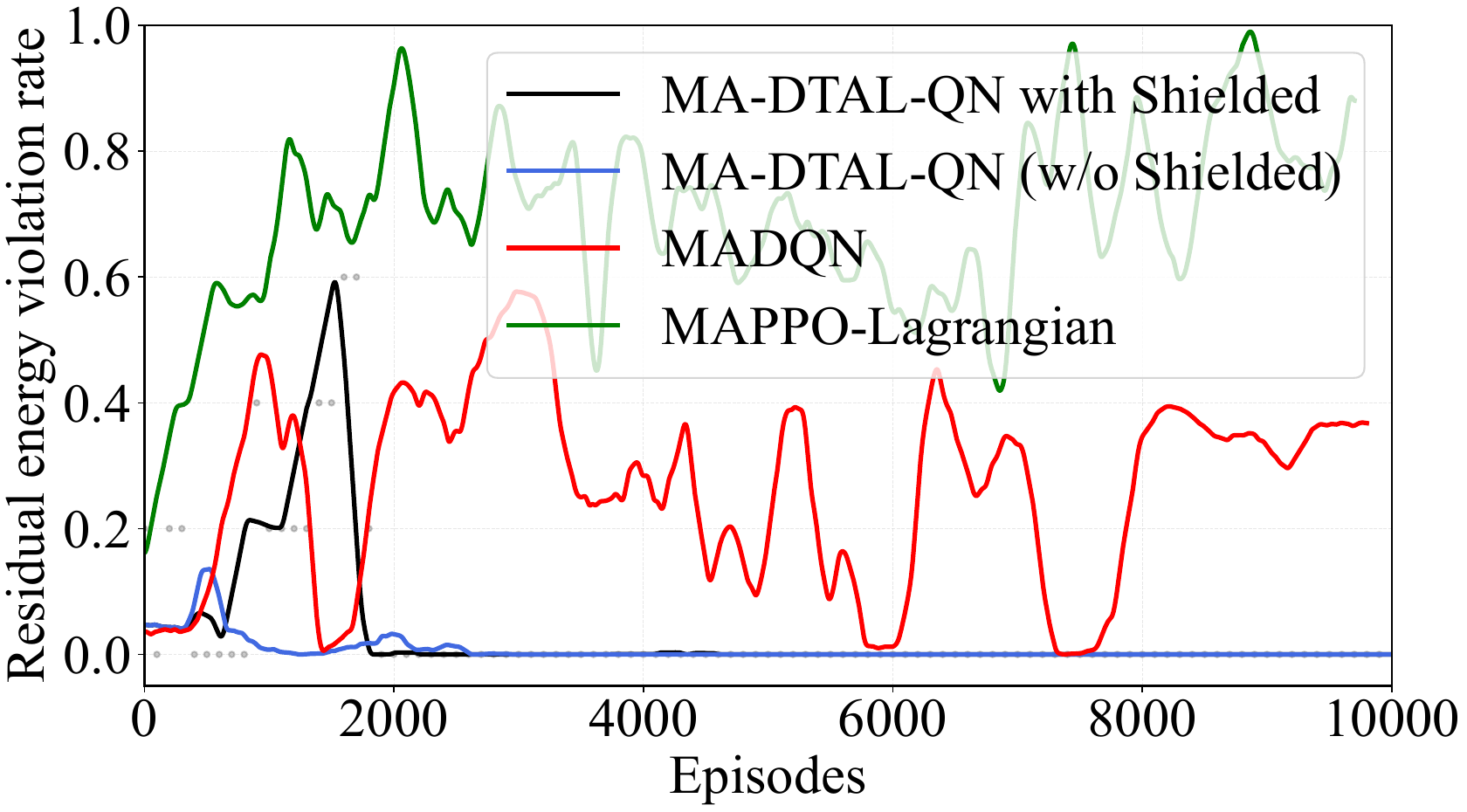}}
\caption{
Energy Constraint Satisfaction: Safe-Deep Q-Learning effectively manages cumulative energy consumption, maintaining sufficient residual energy while balancing communication requirements with operational endurance.}
\label{fig:residual_energy}
\vspace{-10pt}
\end{figure}

Figure \ref{fig:collision_avoidance} delineates the spatial collision rates, explicitly highlighting the efficacy of the proposed hierarchical architecture. By integrating the deterministic Safety Shield, our framework mathematically precludes any inter-UAV distance violations, maintaining an absolute zero-collision rate throughout the entire exploration phase. Conversely, the unshielded ablation eventually learns compliance but suffers inevitable exploratory crashes, whereas the MADQN and MAPPO-Lagrangian baselines chronically fail to enforce safe separation. Furthermore, Figure \ref{fig:residual_energy} evaluates the cumulative residual energy violations. While the baseline methods exhibit erratic and non-convergent energy depletion, both variants of the proposed MA-DTAL-QN rapidly interiorize the energy budget constraints via off-policy augmented Lagrangian updates, achieving strict zero-violation convergence within 2,000 episodes.

\begin{table*}[!tbp]
\caption{Training Performance Comparison: Safe-Deep Q-Learning vs. Baseline Methods}
\label{tab:performance_comparison}
\centering
\begin{tabular}{lcccc}
\hline
\textbf{Metric} & \textbf{MAPPO-Lagrangian} & \textbf{MADQN} & \textbf{Proposed DTAL-QN} & \textbf{Improvement\textsuperscript{*}} \\
\hline
Final U2U Reliability & 91\% &92.3\% & \textbf{100\%} & +8.34\% \\
Final U2R Throughput & 53.74 Mbps &\textbf{69.1} Mbps & 51.2 Mbps & -25.9\% \\
Final Collision Violation Rate & 3\%& 4.7\% & \textbf{0\%} & +100\% \\
Final Residual Energy Satisfaction & 100\% & 88.1\% & \textbf{100\%} & 0\% \\
Convergence Episodes & 20,000 & 8,000 & \textbf{2,300} & +71.25\% \\
\hline
\end{tabular}
\vspace{4pt}
\raggedright
\footnotesize{
\textsuperscript{*}Relative improvement of the proposed method over the best-performing baseline.
}
\end{table*}

Table \ref{tab:performance_comparison} provides a comprehensive evaluation of the proposed DTAL-QN against both unconstrained (MADQN) and constrained (MAPPO-Lagrangian) baselines, illuminating the fundamental trade-off between unconstrained capacity and absolute operational safety.
The results demonstrate significant improvements across all key metrics, particularly in constraint satisfaction where our method achieves near-perfect compliance with all safety requirements while simultaneously enhancing system performance.

\subsubsection{Testing Performance Evaluation}
The testing results in Table \ref{tab:testing_performance} rigorously quantify the fundamental trade-offs between communication capacity and absolute operational safety. While the unconstrained MADQN achieves a nominal throughput of 65 Mbps, this metric is artificially inflated by a severe 19.99\%
 spatial violation rate and unacceptable communication reliability. To provide a rigorous constrained benchmark, we evaluate MAPPO-Lagrangian. 
 Although it significantly improves constraint adherence (reducing violations to 10\% and elevating BRID success to 92\%), it fundamentally fails to guarantee strict safety due to the high gradient variance induced by on-policy penalty updates. 
 In stark contrast, the proposed DTAL-QN (Safe-Deep Q-Learning) attains strictly zero distance violations and 100\% broadcast reliability. To benchmark the performance cost of enforcing this absolute safety, we derived a Genie-aided Upper Bound (76.18 Mbps), which assumes perfect non-causal channel knowledge and zero interference. Operating under strictly causal observations and dense co-channel interference, our proposed framework captures approximately 65.6\% of this idealized theoretical limit, demonstrating a highly effective and pragmatic balance between uncompromised physical safety and networking efficiency.

 As shown in the results, while the baseline MAPPO-Lagrangian reduces constraint violations over time, it exhibits severe sample inefficiency and suboptimal throughput compared to our proposed Safe-Deep Q-Learning. This performance degradation stems from the fundamental nature of on-policy policy gradient methods. In MAPPO, the dynamic scaling of the Augmented Lagrangian penalties drastically alters the variance of the Monte Carlo returns. Consequently, the advantage function is frequently dominated by massive penalty spikes from sporadic safety violations, causing the Actor network's gradient updates to collapse into sub-optimal local minima (e.g., sacrificing all U2R throughput to merely maintain the broadcast constraint).
In stark contrast, our proposed value-based Safe-Deep Q-Learning absorbs these penalties gracefully into the Q-values via off-policy experience replay, naturally decoupling the exploration noise from the harsh Lagrangian constraints and achieving superior convergence stability in the hybrid discrete-continuous action space.

\begin{table*}[!tbp]
\caption{Testing Performance Evaluation over 100 Episodes}
\label{tab:testing_performance}
\centering
\begin{tabular}{lccccc}
\hline
\textbf{Metric} & \textbf{MAPPO-Lagrangian} & \textbf{MADQN} & \textbf{Proposed DTAL-QN} & \textbf{Upper Bound} & \textbf{Improvement\textsuperscript{*}} \\
\hline
Distance Violation Rate & 10\% & 19.99\% &  \textbf{0.00\%} & 0.00\% &\textbf{+100\%} \\
BRID Success Probability & 92\%& 61\% &  \textbf{100\%} & 100\% &\textbf{+8.7\%} \\
DAA Success Probability & 90\% & 54\% & \textbf{100\%} & 100\% & \textbf{+11.1\%} \\
Average U2R Throughput & 65.37 Mbps & 65 Mbps & 50 Mbps & 76.18 Mbps &-23.5\% \\
\hline
\end{tabular}
\vspace{-5pt}
\raggedright
\footnotesize{
\textsuperscript{*}Relative improvement of the proposed method over the best-performing baseline.
}
\end{table*}

\begin{figure}[!tbp]
\centerline{\includegraphics[width=.475\textwidth]{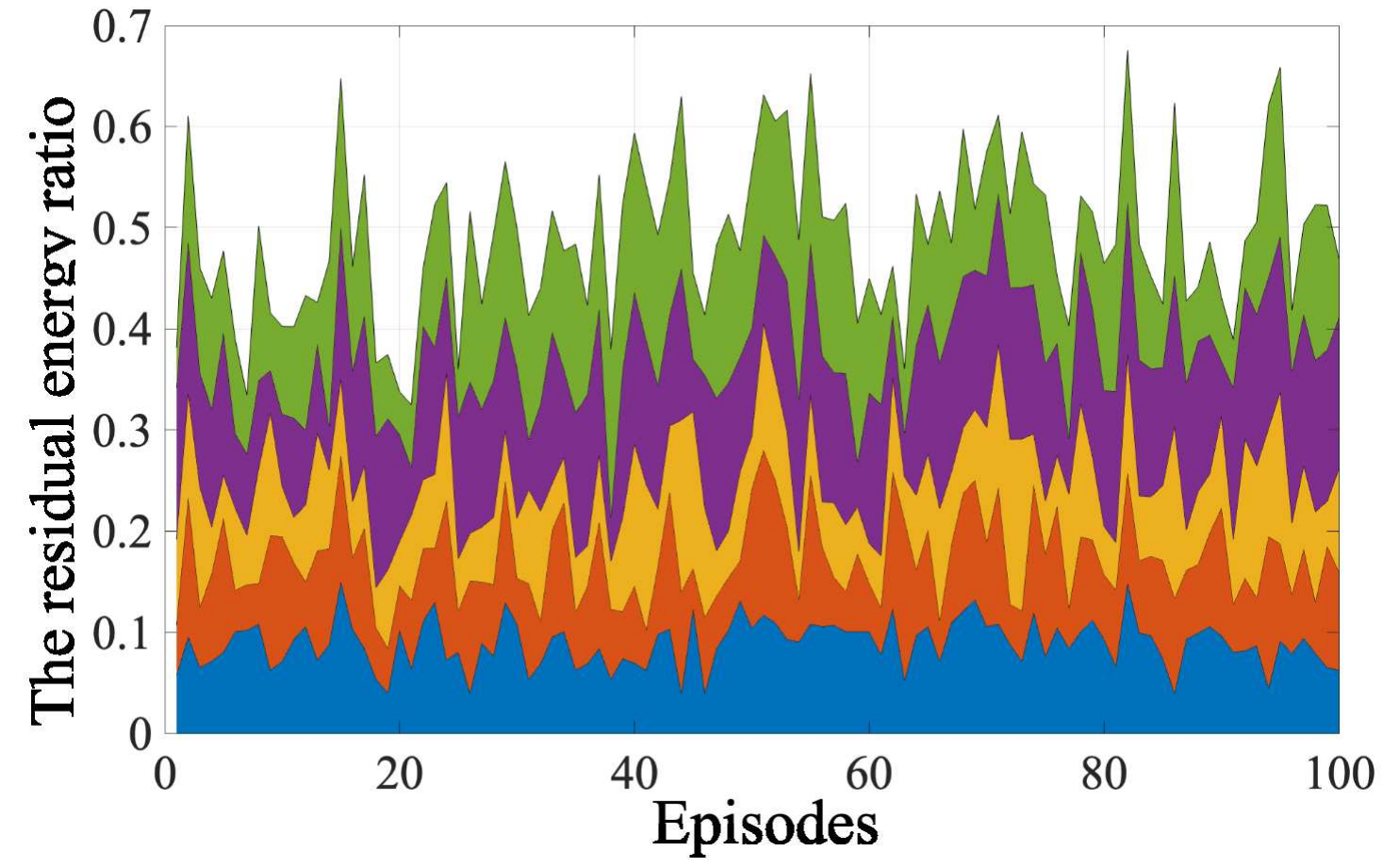}}
\caption{
Residual Energy Safety during Testing: The normalized residual energy of 5 UAVs across 100 testing episodes demonstrates critical energy preservation, with all UAVs maintaining sufficient energy reserves to prevent mission failure due to battery exhaustion. The controlled energy consumption profile ensures operational safety throughout mission duration.
}
\label{fig:residual_energy_testing}
\vspace{-10pt}
\end{figure}

Energy safety performance is critically evaluated to prevent UAV crashes due to battery depletion. Figure \ref{fig:residual_energy} shows effective energy constraint enforcement during training, ensuring UAVs maintain minimum energy levels required for safe operation. Figure \ref{fig:residual_energy_testing} presents testing-phase results, where all 5 UAVs successfully preserve adequate residual energy throughout 100 episodes. The energy profiles demonstrate that while energy is progressively consumed during missions, sufficient reserves are maintained to prevent catastrophic failures, validating the algorithm's capability to enforce this life-critical constraint.

\subsubsection{Scalability and Zero-Shot Generalization}
A critical concern in multi-agent systems is the combinatorial explosion of the joint action space, $\mathcal{O}(|A^{N_u}|)$. The proposed CTDE architecture inherently resolves this bottleneck. 
During deployment, the execution is fully decentralized; thus, the action space per agent remains strictly constant (i.e., $|A_n|=255$), ensuring computational scalability regardless of the swarm size.
To empirically validate this architectural scalability, we conducted a zero-shot generalization experiment. We scaled the swarm size from $N_u=6$ to $N_u=10$ while strictly reusing the neural network weights pre-trained solely on the $5$-UAV environment. 
To prevent the observation space from expanding, we implemented a localized K-Nearest Neighbors Field of View, wherein each agent perceives only its closest neighbors. The zero-shot performance, averaged over 50 independent episodes, is summarized in Table \ref{tab:scalability}.

The results in Table \ref{tab:scalability} provide two critical insights regarding the framework's scalability. First, the localized spatial observation design demonstrates highly robust zero-shot generalization. Across unseen densities from $N_u=6$ to $9$, the policy maintains a strictly zero raw violation probability, indicating that the agents have thoroughly interiorized the spatial safety boundaries during the base training phase. Even at the maximum tested density of $N_u=10$, the intended violation probability only marginally increases to $0.001$. Crucially, these isolated deviations are deterministically intercepted by the execution-layer Safety Shield, thereby guaranteeing absolute physical collision avoidance regardless of the swarm density.
Second, the progressive degradation in U2R throughput (from $39.78$ Mbps to $3.49$ Mbps) and broadcast reliability is a direct consequence of electromagnetic physical limits rather than an algorithmic scaling failure. Because the number of active transmitters strictly exceeds the available orthogonal spectrum resources ($5$ sub-bands), the system encounters severe unmodeled co-channel interference. Consequently, the decentralized agents adopt conservative spatial maneuvers to prioritize physical safety over instantaneous capacity. This confirms that the proposed framework remains computationally and spatially scalable, with networking performance appropriately bottlenecked solely by the finite physical spectrum availability.
\begin{table*}[htbp]
\renewcommand{\arraystretch}{1.2}
\caption{Zero-Shot Generalization Performance under Scaled Swarm Densities (Pre-trained on $N_u=5$)}
\label{tab:scalability}
\centering
\begin{tabular}{c c c c}
\hline
\hline
\textbf{Swarm Size ($N_u$
)} & \textbf{U2R Rate (Mbps)} & \textbf{Broadcast Prob.} & \textbf{Raw Violation Prob.} \\
\hline
6 & 39.78 & 0.84 & 0.00
 \\
7 & 26.21 & 0.74 & 0.00
 \\
8 & 18.74 & 0.68 & 0.00
 \\
9 & 9.36 & 0.56 & 0.00
 \\
10 & 3.49 & 0.52 & 0.001
 \\
\hline
\hline
\end{tabular}
\end{table*}

\section{Study case 2: Post-Disaster Emergency Communications Under Outage Constraints}

\label{sec:studycase2}
\begin{figure}[!tbp]
\centerline{\includegraphics[width=.475\textwidth]{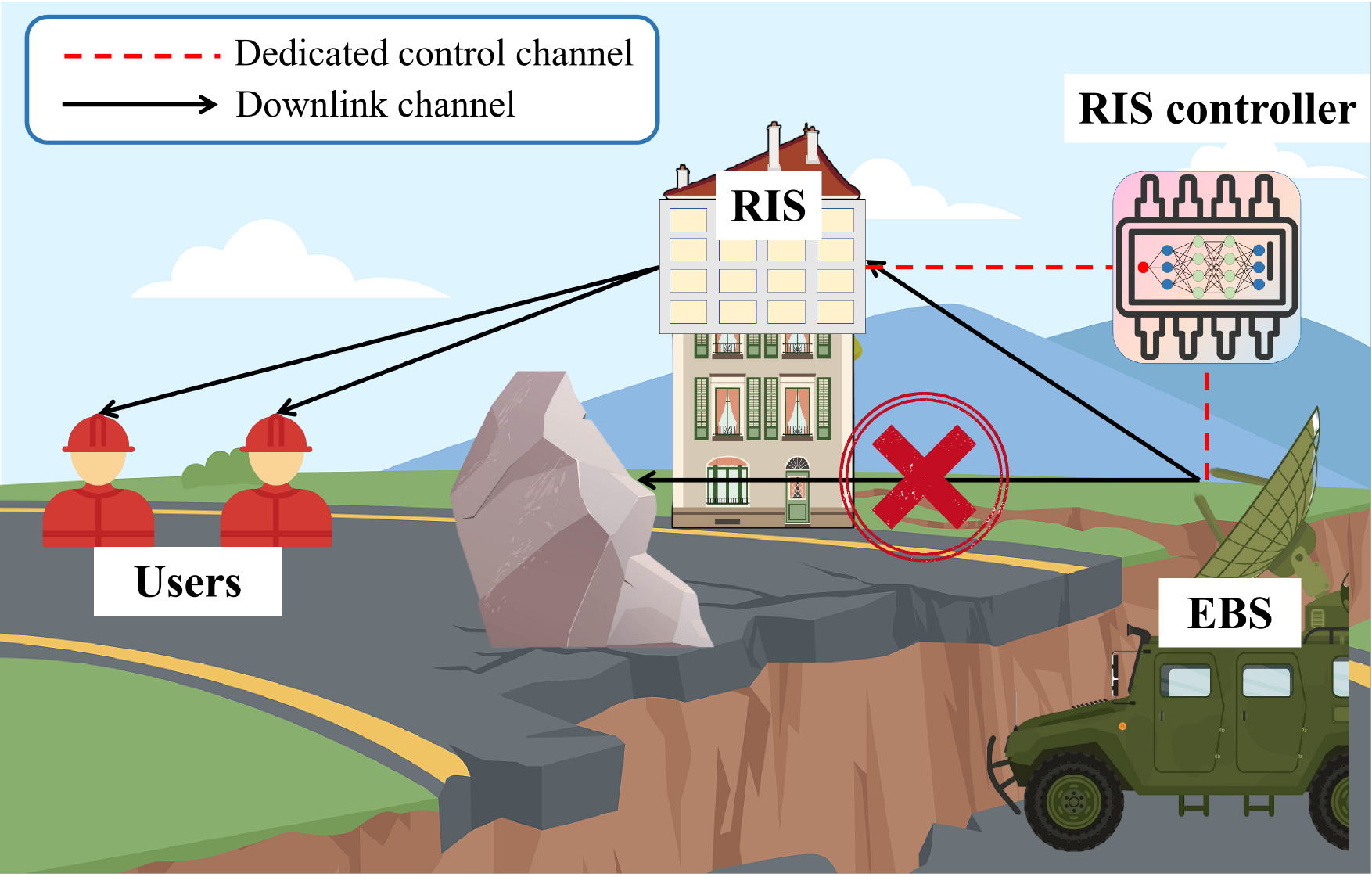}}
\caption{
RIS-assisted post-disaster emergency communication system, where an Emergency Base Station (EBS) utilizes a deployed RIS to bypass physical obstructions (e.g., collapsed buildings) and establish reliable links with rescue teams.
}
\label{Fig:RIS_system}
\vspace{-10pt}
\end{figure}

In the aftermath of catastrophic events (e.g., earthquakes or hurricanes), terrestrial communication infrastructure is often severely compromised, creating ``blind spots” that hinder rescue operations.
Re-establishing connectivity in such scenarios is a paramount safety-critical task governed by three fundamental requirements: 
1) \textbf{Ubiquitous Coverage:} ensuring signal reachability to rescue teams and trapped victims in shadowed areas; 
2) \textbf{Service Reliability:} minimizing outage probabilities to maintain continuous command-and-control links; and 
3) \textbf{Hardware Resilience:} ensuring deployed equipment operates strictly within physical feasibility margins to prevent device failure under extreme conditions.

To address the coverage challenge, RIS have emerged as a vital technology for emergency network recovery.
By reshaping the electromagnetic environment, RIS can create virtual LoS links to bypass physical obstructions such as collapsed buildings or debris. 
However, integrating RIS introduces complex hardware-intrinsic constraints. 
Specifically, owing to the passive nature of reflecting elements (which lack active RF chains), the phase shift configuration must strictly satisfy the \textit{unit-modulus constraint}. 
Violating this physical boundary renders the configuration invalid, leading to immediate communication failure. 
Consequently, this case study applies our Safe-Deep Q-Learning framework to optimize robust beamforming in an RIS-aided emergency network, ensuring maximum coverage and reliability while strictly enforcing the hardware's physical operational safety.

\subsection{System Model}
We consider a downlink emergency communication scenario where a deployable Emergency Base Station (EBS) equipped with $M_t$ antennas serves $M_u$ single-antenna rescue teams. 
Due to severe environmental obstruction (e.g., urban rubble), the direct links between the EBS and rescue teams are assumed to be blocked. To restore connectivity, an RIS with $M_r$ passive reflecting elements is deployed to establish a cascaded communication channel.

The channel from the EBS to the RIS is denoted by $\mathbf{G} \in \mathbb{C}^{M_t \times M_r}$, and the channel from the RIS to user $m_u$ is denoted by $\mathbf{h}_{r,m_u}^H \in \mathbb{C}^{1 \times M_r}$.
The reflection coefficient matrix of the RIS is modeled as a diagonal matrix $\mathbf{\Theta} = \text{diag}(e^{j\vartheta_1}, \ldots, e^{j\vartheta_{M_r}})$, where $\vartheta_{m_r} \in [0, 2\pi)$ represents the phase shift of the $m_r$-th element.
Crucially, to ensure hardware resilience and physical feasibility, each element must satisfy the unit-modulus constraint $|e^{j\vartheta_{m_r}}| = 1$.

The EBS employs linear precoding with beamforming vectors $\mathbf{w}_{m_u} \in \mathbb{C}^{M_t \times 1}$ for each rescue team $m_u$. The transmitted signal is $\mathbf{x} = \sum_{m_u=1}^{M_u} \mathbf{w}_{m_u} s_{m_u}$, where $s_{m_u}$ is the critical data symbol. The received SINR at rescue team $k$ is given by:
\begin{equation}
\text{SINR}_k = \frac{ \left| \left( \mathbf{h}_{r,m_u}^H \mathbf{\Theta} \mathbf{G} \right) \mathbf{w}_{m_u} \right|^2 }{ \sum_{m_u' \neq m_u} \left| \left( \mathbf{h}_{r,m_u}^H \mathbf{\Theta} \mathbf{G} \right) \mathbf{w}_{m_u'} \right|^2 + \sigma_{m_u}^2 },
\end{equation}
where $\sigma_{m_u}^2$ denotes the noise power.

\subsection{Problem Formulation}
The objective is to minimize the total power consumption of the battery-constrained EBS while ensuring strict QoS reliability for rescue operations and adhering to the physical hardware constraints of the RIS.
The problem is formulated as:
\begin{equation}
\begin{aligned}
& \underset{\mathbf{W}, \boldsymbol{\vartheta}}{\min} \quad \sum_{m_u=1}^{M_u} \| \mathbf{w}_{m_u} \|^2 \\
& \text{s.t.} \quad \text{SINR}_{m_u} \geq \Gamma_{m_u}, \quad {m_u} = 1, \ldots, {M_u}, \\
& \quad \quad |e^{j\vartheta_{m_r}}| - 1 = 0, \quad {m_r} = 1, \ldots, {M_r},
\end{aligned}
\label{eq:ris_problem_formulation}
\end{equation}
where $\Gamma_{m_u}$ represents the minimum SINR threshold required for reliable emergency communication. The equality constraint $|e^{j\vartheta_{m_r}}| - 1 = 0$ represents the instantaneous hard constraint enforced by the RIS hardware physics, which creates a non-convex optimization landscape suitable for our proposed Safe-Deep Q-Learning approach.

\begin{table*}[!tbp]
\caption{Testing Performance Comparison for RIS-assisted Beamforming System (100 Episodes)}
\label{tab:ris_testing_comparison}
\centering 
\begin{tabular}{lccc}
\hline
\textbf{Metric} & \textbf{Standard DQN} & \textbf{Safe-Deep Q-Learning (ALDQN)} & \textbf{Improvement} \\
\hline
Energy Cost (Watts) & \textbf{6.7513} & 9.9626 & \textbf{-47.5\%} \\
Feasible Probability & 81\% & \textbf{100\%} & \textbf{+23.5\%} \\
\hline
\end{tabular}
\vspace{-5pt}
\end{table*}

\subsection{State-action representation}
The action space for the RIS-assisted system is designed as a discrete combinatorial space that jointly optimizes transmit power allocation and RIS phase configuration. Each action encodes:
\begin{itemize}
    \item \textbf{Power Allocation}: Each user's transmit power is selected from $\mathcal{Z}_p$ discrete levels in dBm scale.
    \item \textbf{RIS Phase Configuration}: The RIS is partitioned into $\mathcal{Z}_{\text{blocks}} = \mathcal{Z}_{\text{blocks}}^x \times \mathcal{Z}_{\text{blocks}}^z$ sub-blocks, with each block sharing a common phase shift selected from an 8-point DFT codebook $\{\vartheta_1, \vartheta_2, \dots, \vartheta_8\}$ covering $[0, 2\pi)$.
\end{itemize}
The complete action is encoded as a single integer $a \in [0, (\mathcal{Z}_p)^{M_u} \times 8^{\mathcal{Z}_{\text{blocks}}}-1]$. 
Action decoding involves modular arithmetic to extract power indices and phase shift indices for subsequent system configuration.

The observation space captures the complete electromagnetic environment through comprehensive channel state information, comprising the EBS-RIS channel matrix, and the RIS-user channels. 
These matrices are flattened and concatenated into a continuous vector representation by separating the real and imaginary components. This rich perceptual representation enables the learning agent to discern spatial channel characteristics and fading patterns essential for orchestrating joint beamforming and reflection optimization in dynamic propagation environments.

\subsection{Evaluation results}
\subsubsection{Experimental Setup and System Configuration}
The RIS-assisted beamforming system comprises an access point with $M_t = 8$ antennas, an intelligent reflecting surface with $M_r = 64$ configurable passive elements, and $M_u = 4$ single-antenna users. The direct EBS-user links are assumed to be blocked, necessitating the use of RIS for signal reflection and beamforming. The system operates in the 28 GHz millimeter-wave band with 100 MHz bandwidth, representative of next-generation wireless deployments.

The channel model incorporates both large-scale path loss and small-scale Rician fading with $K$-factor of 10 dB. The EBS-RIS channel $\mathbf{G}$ follows a geometric channel model with 5 scattering clusters, while the RIS-user channels $\mathbf{h}_{r,m_u}$ experience independent Rician fading. User SINR requirements are set to $\Gamma_{m_u} = 10$ dB for all users, representing typical quality-of-service thresholds for high-data-rate applications.

\begin{figure}[!tbp]
\centerline{\includegraphics[width=.475\textwidth]{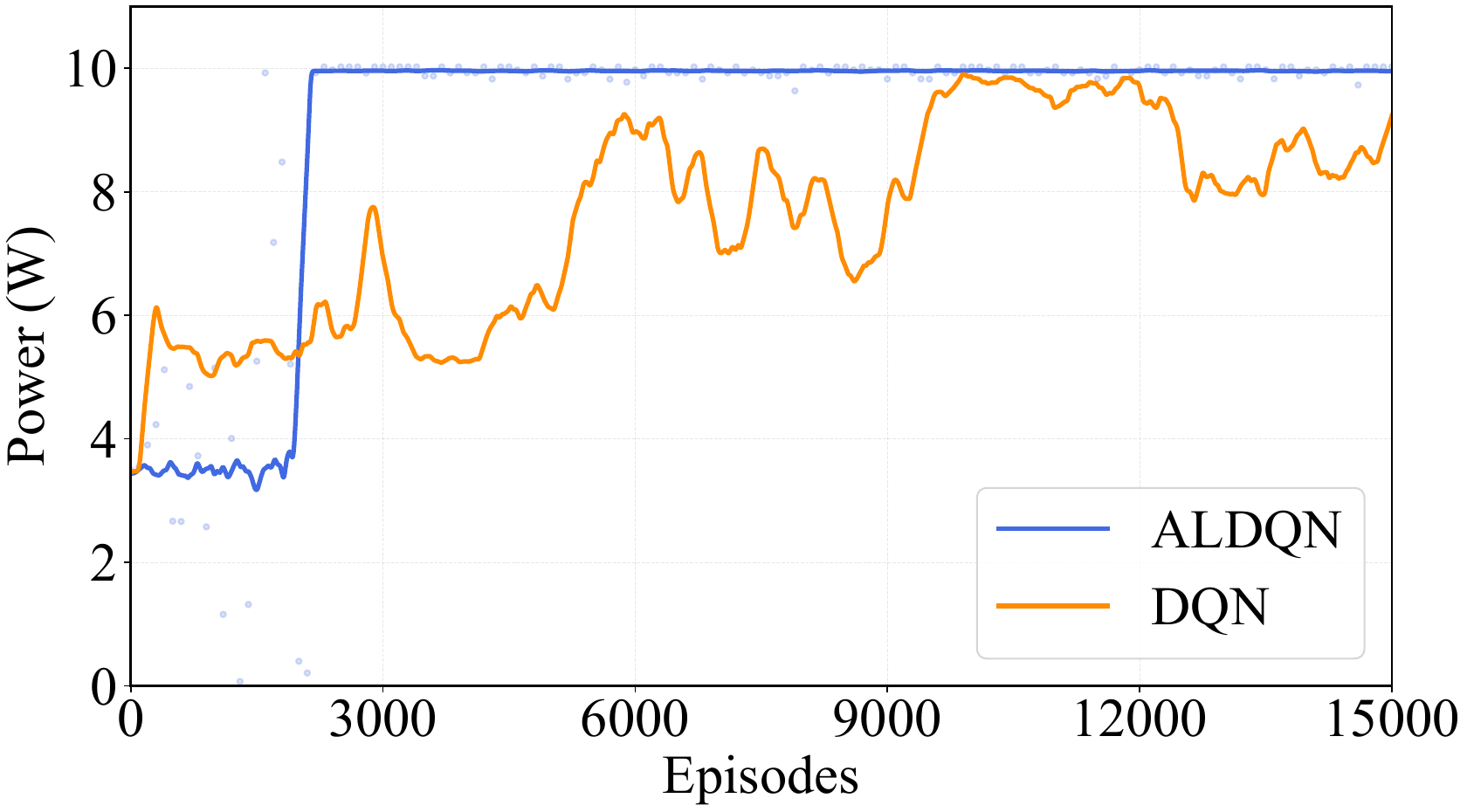}}
\caption{
Training Phase Power Consumption under SINR Constraints: Safe-Deep Q-Learning demonstrates adaptive power management during training, with increasing power consumption as the algorithm learns to satisfy stringent SINR QoS requirements. The progressive power escalation reflects the necessary trade-off to maintain $\Gamma_{m_u} = 10$ dB SINR thresholds for all users, validating the algorithm's capability to enforce QoS constraints through appropriate resource allocation.
}
\label{fig:energy_efficiency}
\vspace{-10pt}
\end{figure}

\subsubsection{Training Phase Constraint Satisfaction Analysis}
Figure \ref{fig:energy_efficiency}  illustrates the power consumption during training. The Safe-Deep Q-Learning demonstrates a more conservative and stable power adjustment strategy compared to the erratic fluctuations of the baseline DQN. 
The proposed method gradually adjusts power allocation to satisfy the stringent user-specific SINR constraints ($\Gamma_{m_u} = 10$ dB).

\begin{figure}[!tbp]
\centerline{\includegraphics[width=.475\textwidth]{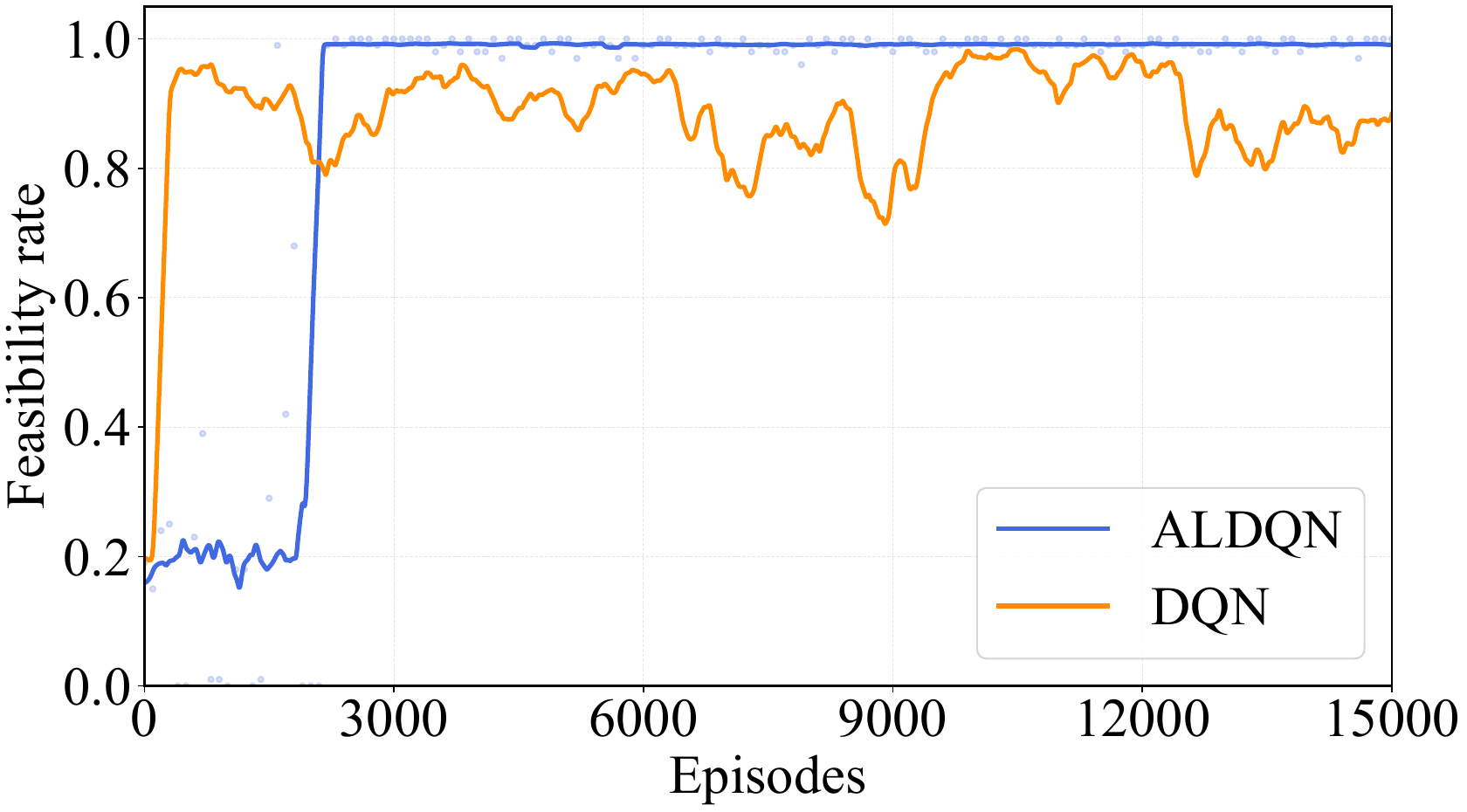}}
\caption{
Training Phase Feasibility Convergence: Safe-Deep Q-Learning achieves rapid convergence to 100\% feasibility rate during training, significantly outperforming conventional alternating optimization. The algorithm demonstrates stable learning progress in handling the nonconvex joint beamforming and phase shift optimization, with feasibility performance stabilizing after approximately 2,500 training episodes while maintaining SINR constraints through appropriate power allocation.
}
\label{fig:ris_feasibility}
\vspace{-10pt}
\end{figure}

Figure \ref{fig:ris_feasibility}  highlights the feasibility convergence. 
The proposed algorithm achieves rapid improvement, reaching near 100\% feasibility by episode 2,500. Conversely, the baseline DQN saturates at a lower feasibility rate (approx. 80-85\%), failing to consistently find solutions that satisfy the nonconvex joint beamforming and phase-shift constraints.

The high feasibility rate is achieved through intelligent coordination between beamforming design and power allocation, where the algorithm learns to allocate sufficient power to overcome channel impairments and interference while optimizing the RIS phase shifts for signal enhancement.

\subsubsection{Testing Performance Evaluation}
To comprehensively evaluate the practical deployment performance of Safe-Deep Q-Learning in RIS-assisted systems, we conducted extensive testing over 100 episodes under realistic operational conditions.
Table \ref{tab:ris_testing_comparison} presents the quantitative testing results. Our approach achieves 100\% feasibility probability, a significant improvement over the baseline's 81\%. 
However, ensuring this strict compliance imposes an operational cost: the proposed method incurs a higher energy cost (9.96W vs 6.75W), representing a 47.5\% increase in power consumption compared to the baseline.
This result clarifies that the baseline 'saves' energy only by frequently failing to meet the required QoS targets.
The proposed Safe-Deep Q-Learning correctly identifies that higher power expenditure is required to guarantee the 10 dB SINR target in this challenging blocked-link environment.

\section{Conclusions}
\label{conclusions}
This article proposed Safe-Deep Q-Learning, a constrained control framework designed to enable safety-critical wireless communications under stochastic dynamics and mixed-integer nonconvexity. By integrating augmented Lagrangian methods with a dual-timescale update mechanism, the proposed algorithm theoretically guarantees convergence while enforcing both instantaneous and long-term constraints. Evaluations on multi-UAV networks and RIS-assisted systems demonstrate that the framework achieves 100\% constraint satisfaction (e.g., zero collisions, strict QoS compliance). Unlike standard RL baselines that maximize raw performance at the risk of system failure, our approach effectively prioritizes operational safety, establishing a robust foundation for trustworthy 6G autonomous networks.

\ifCLASSOPTIONcaptionsoff
  \newpage
\fi

\bibliographystyle{IEEEtran}
\bibliography{ref.bib}

@ARTICLE{10051712,
  author={Peng, Haoran and Wang, Li-Chun},
  journal={IEEE Trans. Wireless Commun.}, 
  title={Energy Harvesting Reconfigurable Intelligent Surface for {UAV} Based on Robust Deep Reinforcement Learning}, 
  year={2023},
  volume={22},
  number={10},
  pages={6826--6838},
  month={Oct.}}

@INPROCEEDINGS{9683088,
  author={Li, Jingqi and Fridovich-Keil, David and Sojoudi, Somayeh and Tomlin, Claire J.},
  booktitle={Proc. IEEE Conf. Decis. Control.}, 
  title={Augmented {Lagrangian} Method for Instantaneously Constrained Reinforcement Learning Problems}, 
  year={2021},
  pages={2982--2989},
month={Dec.},
address = {Austin, TX}
}

@ARTICLE{10443575,
  author={Zhang, Tong and Gou, Yu and Liu, Jun and Song, Shanshan and Yang, Tingting and Cui, Jun-Hong},
  journal={IEEE Trans. Mob. Comput.}, 
  title={Joint Link Scheduling and Power Allocation in Imperfect and Energy-Constrained Underwater Wireless Sensor Networks}, 
  year={2024},
  volume={23},
  number={10},
  pages={9863--9880},
  month={Oct.}}

@ARTICLE{10700695,
  author={Ning, Zhaolong and Ji, Hongjing and Wang, Xiaojie and Ngai, Edith C. H. and Guo, Lei and Liu, Jiangchuan},
  journal={IEEE Trans. Mob. Comput.}, 
  title={Joint Optimization of Data Acquisition and Trajectory Planning for {UAV}-Assisted Wireless Powered Internet of Things}, 
  year={2025},
  volume={24},
  number={2},
  pages={1016--1030},
  month={Feb.}}

@ARTICLE{11087615,
  author={Gunes Reyhan, Berire and Coleri, Sinem},
  journal={IEEE Trans. Commun.}, 
  title={Safe Deep Reinforcement Learning for Resource Allocation With Peak Age of Information Violation Guarantees}, 
  year={2025},
  volume={73},
  number={12},
  pages={14197--14211},
  month={Dec.}}

@ARTICLE{10974474,
  author={Zhu, Bincheng and Huang, Liang and Chi, Kaikai and Alharbi, Abdullah and Yu, Keping and Guizani, Mohsen},
  journal={IEEE Trans. Wireless Commun.}, 
  title={Enhancing Energy Efficiency in Wireless-Powered {MEC} Systems Through {Lyapunov}-Guided Deep Reinforcement Learning}, 
  year={2025},
  volume={24},
  number={9},
  pages={7563--7580},
  month={Sep.}}

@ARTICLE{10529729,
  author={Pourkabirian, Azadeh and Kordafshari, Mohammad Sadegh and Jindal, Anish and Anisi, Mohammad Hossein},
  journal={IEEE Commun. Stand. Mag.}, 
  title={A Vision of {6G} {URLLC}: Physical-Layer Technologies and Enablers}, 
  year={2024},
  volume={8},
  number={2},
  pages={20--27},
  month={Jun.}}

@ARTICLE{10675394,
  author={Gu, Shangding and Yang, Long and Du, Yali and Chen, Guang and Walter, Florian and Wang, Jun and Knoll, Alois},
  journal={IEEE Trans. Pattern Anal. Mach. Intell.}, 
  title={A Review of Safe Reinforcement Learning: Methods, Theories, and Applications}, 
  year={2024},
  volume={46},
  number={12},
  pages={11216--11235},
  month={Dec.}}

@ARTICLE{10269137,
  author={Lyu, Ting and Xu, Haitao and Zhang, Long and Han, Zhu},
  journal={IEEE Internet Things J.}, 
  title={Source Selection and Resource Allocation in Wireless-Powered Relay Networks: An Adaptive Dynamic Programming-Based Approach}, 
  year={2024},
  volume={11},
  number={5},
  pages={8973--8988},
month={Mar.}}

@ARTICLE{9968223,
  author={Rostampoor, Javane and Adve, Raviraj S.},
  journal={IEEE Trans. Commun.}, 
  title={Optimizing Caching in a {C-RAN} With a Hybrid Millimeter-Wave/Microwave Fronthaul Link via Dynamic Programming}, 
  year={2023},
  volume={71},
  number={2},
  pages={923--934},
month={Feb.}}

@ARTICLE{9770190,
  author={Zhang, Yichi and Zhao, Haitao and Wei, Jibo and Zhang, Jiao and Flanagan, Mark F. and Xiong, Jun},
  journal={IEEE Trans. Cog. Commun. Netw.}, 
  title={Context-Based Semantic Communication via Dynamic Programming}, 
  year={2022},
  volume={8},
  number={3},
  pages={1453-1467},
month={Sep.}}

@ARTICLE{10666879,
  author={Lu, Binbin and Wu, Yuan and Qian, Liping and Zhou, Sheng and Zhang, Haixia and Lu, Rongxing},
  journal={IEEE Trans. Netw. Serv. Manag.}, 
  title={Multi-Agent {DRL}-Based Two-Timescale Resource Allocation for Network Slicing in {V2X} Communications}, 
  year={2024},
  volume={21},
  number={6},
  pages={6744--6758},
month={Dec.}}

@ARTICLE{10605908,
  author={Kaneko, Megumi and Dinh, Thi Ha Ly and Shnaiwer, Yousef N. and Kawamura, Kenichi and Murayama, Daisuke and Takatori, Yasushi},
  journal={IEEE Wireless Commun. Lett.}, 
  title={A Multi-Agent Risk-Averse Reinforcement Learning Method for Reliability Enhancement in Sub-{6GHz}/mmWave Mobile Networks}, 
  year={2024},
  volume={13},
  number={10},
  pages={2657--2661},
month={Oct.}}

@ARTICLE{8771176,
  author={Huang, Liang and Bi, Suzhi and Zhang, Ying-Jun Angela},
  journal={IEEE Trans. Mob. Comput.}, 
  title={Deep Reinforcement Learning for Online Computation Offloading in Wireless Powered Mobile-Edge Computing Networks}, 
  year={2020},
  volume={19},
  number={11},
  pages={2581--2593},
month={Nov.}}

@ARTICLE{10636057,
  author={Luo, Chuang and Jiang, Weiheng and Niyato, Dusit and Ding, Zhiguo and Li, Jingfu and Xiong, Zehui},
  journal={IEEE Trans. Wireless Commun.}, 
  title={Optimization and {DRL}-Based Joint Beamforming Design for Active-{RIS} Enabled Cognitive Multicast Systems}, 
  year={2024},
  volume={23},
  number={11},
  pages={16234--16247},
month={Nov.}}

@ARTICLE{10943268,
  author={Long, Yusi and Gong, Shimin and Sun, Sumei and Lee, Gary C. F. and Li, Lanhua and Niyato, Dusit},
  journal={IEEE Trans. Wireless Commun.}, 
  title={Lyapunov-Guided Deep Reinforcement Learning for Semantic-Aware {AoI} Minimization in UAV-Assisted Wireless Networks}, 
  year={2025},
  volume={24},
  number={8},
  pages={6351--6364},
  month={Aug.}}

@ARTICLE{8854339,
  author={Yan, Jia and Bi, Suzhi and Zhang, Ying Jun and Tao, Meixia},
  journal={IEEE Trans. Wireless Commun.}, 
  title={Optimal Task Offloading and Resource Allocation in Mobile-Edge Computing With Inter-User Task Dependency}, 
  year={2020},
  volume={19},
  number={1},
  pages={235--250},
  month={Jan.}}

@ARTICLE{10746553,
  author={Zhong, Yi and Chen, Zhuoling and Han, Tao and Ge, Xiaohu},
  journal={IEEE Trans. Commun.}, 
  title={Spatio-Temporal Interference Correlation: Influence of Deployment Patterns and Traffic Dynamics}, 
  year={2025},
  volume={73},
  number={5},
  pages={3199--3213},
  month={May}}

@ARTICLE{11074719,
  author={Su, Tong and Wu, Tong and Zhao, Junbo and Scaglione, Anna and Xie, Le},
  journal={Proc. IEEE}, 
  title={A Review of Safe Reinforcement Learning Methods for Modern Power Systems}, 
  year={2025},
  volume={113},
  number={3},
  pages={213--255},
  month={Mar.}}

@article{mnih2015human,
  title={Human-level control through deep reinforcement learning},
  author={Mnih, Volodymyr and Kavukcuoglu, Koray and Silver, David and Rusu, Andrei A and Veness, Joel and Bellemare, Marc G and Graves, Alex and Riedmiller, Martin and Fidjeland, Andreas K and Ostrovski, Georg and others},
  journal={nature},
  volume={518},
  number={7540},
  pages={529--533},
  year={2015},
  month={Feb.}
}

@INPROCEEDINGS{10.5555/3454287.3454966,
author = {Paternain, Santiago and Chamon, Luiz F. O. and Calvo-Fullana, Miguel and Ribeiro, Alejandro},
title = {Constrained reinforcement learning has zero duality gap},
year = {2019},
address = {Red Hook, NY},
booktitle = {Proc. 33rd Neural Inf. Process. Syst. (NIPS)},
pages = {7555--7565}
}

@ARTICLE{11017513,
  author={Sun, Ruijin and Cheng, Nan and Li, Changle and Quan, Wei and Zhou, Haibo and Wang, Ying and Zhang, Wei and Shen, Xuemin},
  journal={IEEE Commun. Surv. Tutor.}, 
  title={A Comprehensive Survey of Knowledge-Driven Deep Learning for Intelligent Wireless Network Optimization in {6G}}, 
  year={2026},
  volume={28},
  number={1},
  pages={1099--1135},
  month={May}}

@ARTICLE{10970466,
  author={Chen, Bingwen and Ling, Xintong and Cao, Weihang and Wang, Jiaheng and Ding, Zhi},
  journal={IEEE J. Sel. Areas Commun.}, 
  title={Analysis of Channel Uncertainty in Trusted Wireless Services via Repeated Interactions}, 
  year={2025},
  volume={43},
  number={6},
  pages={2248--2265},
  month={Jun.}}

@ARTICLE{10974617,
  author={Cai, Jianghui and Chen, Bujia and Zhang, Mian and Wen, Jie and Cui, Zhihua and Chen, Jinjun},
  journal={IEEE Internet Things J.}, 
  title={A Fairness-Aware Resource Management Model With Many-Objective Optimization in Uncertain Resource-Constrained Internet of Vehicles}, 
  year={2025},
  volume={12},
  number={14},
  pages={26718--26729},
  month={Jul.}}

@ARTICLE{10877725,
  author={Ben Amar, Eya and Ben Rached, Nadhir and Tempone, Raúl and Alouini, Mohamed-Slim},
  journal={IEEE Trans. Wireless Commun.}, 
  title={Stochastic Differential Equations for Performance Analysis of Wireless Communication Systems}, 
  year={2025},
  volume={24},
  number={5},
  pages={4040--4054},
  month={May}}

@ARTICLE{11031135,
  author={Im, Hyeon-Seong and Kim, Kyu-Yeong and Lee, Si-Hyeon},
  journal={IEEE Trans. Veh. Technol.}, 
  title={Trajectory Optimization for Cellular-Enabled {UAV} With Connectivity and Battery Constraints}, 
  year={2025},
  volume={74},
  number={11},
  pages={17812--17828},
  month={Nov.}}

\end{document}